\documentclass[11pt]{article}
\usepackage[vcentermath]{youngtab}
\usepackage[font=small]{caption}
\usepackage{tikz}
\usepackage{cite}
\usepackage{amsthm}
\usepackage{amssymb}
\usepackage{geometry}
\usepackage[english]{babel}
\usepackage[utf8]{inputenc}
\usepackage[T1]{fontenc}
\usepackage{indentfirst}
\usepackage{amsmath}
\usepackage{amssymb}
\usepackage{graphicx}
\usepackage{psfrag}
\usepackage{bbold}
\usepackage{proof}
\usepackage{lmodern}

\allowhyphens

\newcommand{\complex}{\mathbb{C}}

\newcommand{\valos}{\mathbb{R}}
\newcommand{\eps}{\varepsilon}

\newcommand{\ii}{^{(i)}}

\newcommand{\ordo}{\mathcal{O}}
\newcommand{\Tr}{\text{Tr}}
\newcommand{\physstr}{\mathcal{P}}

\newtheorem{thm}{Theorem}
\newtheorem{conj}{Conjecture}

\newtheorem{lemma}{Lemma}

\newcommand{\vev}[1]{\left\langle #1 \right\rangle}
\newcommand{\ket}[1]{{\left|#1\right\rangle}}
\newcommand{\bra}[1]{{\left\langle #1\right|}}
\newcommand{\braket}[1]{{\left\langle #1\right\rangle}}
\newcommand{\skalarszorzat}[2]{{\langle #1 | #2 \rangle}}

\def\ii{\mathrm{i}}


\usepackage{ifpdf}

\ifpdf
\usepackage{epstopdf}
\usepackage[pdftex,colorlinks,urlcolor=blue,citecolor=blue,linkcolor=blue]{hyperref}
\else
\usepackage[hypertex,colorlinks,urlcolor=blue,citecolor=blue,linkcolor=blue]{hyperref}
\fi
\usepackage{cleveref}
\title{Generalized Gibbs Ensemble and string-charge relations\\ in
  nested Bethe Ansatz}
\author{Gy\"{o}rgy Z. Feh\'{e}r$^{1}$, Bal\'{a}zs Pozsgay$^{2}$  \bigskip \\
\parbox{0.7\textwidth}{\center  \small
 $^1$\textit{BME Statistical Field Theory Research Group\\
   1111 Budapest, Budafoki \'ut 8, Hungary}\\
 \bigskip
 $^2$\textit{
MTA-BME Quantum Dynamics and Correlations Research Group,
    Department of Theoretical Physics,\\ Budapest University
of Technology and Economics,\\ 1521 Budapest, Hungary}
\small\date{\small\today}   
}}

\begin{document}
\numberwithin{equation}{section}
\maketitle
\abstract{The non-equilibrium steady states of integrable models are
  believed to be described by the Generalized Gibbs Ensemble (GGE), which
  involves all local and quasi-local conserved charges of the model. In
  this work we investigate integrable lattice models solvable by the
  nested Bethe Ansatz, with group symmetry $SU(N)$, $N\ge 3$. In these
  models the Bethe Ansatz involves various types of Bethe rapidities
  corresponding to the ``nesting'' procedure, describing the internal
  degrees of freedom for the excitations. We show that a complete set
  of charges for the GGE can be obtained from the known fusion hierarchy of
  transfer matrices. The resulting charges are quasi-local in a
  certain regime in rapidity space, and they completely fix the
  rapidity distributions of each string type from each nesting
  level.
  }

\section{Introduction}

One of the central problems in theoretical physics is the
connection between the fundamental laws of quantum mechanics and the
various classical theories describing physics on macroscopic scales. A
particularly interesting question is equilibration and thermalization
of closed quantum systems, i.e. the emergence of statistical
mechanics from the unitary time evolution dictated by the
Schr\"odinger equation. This problem has attracted interest since the
1930's, and significant understanding has been achieved in the last 15
years (for reviews, see \cite{rigol-eth,Silva-quench-colloquium}).
Furthermore, special attention has been
devoted to those systems which do not thermalize, and one class of
such systems are the integrable models.

One dimensional exactly solvable models are known to possess an
infinite number of mutually commuting conserved charges. The resulting
conservation laws prevent the integrable systems from thermalization
to standard Gibbs ensembles. Instead, the idea of the Generalized
Gibbs Ensemble (GGE) was put forward in \cite{rigol-gge,rigol-2}. The
GGE is analogous to the canonical Gibbs ensemble: it is built on the
maximum entropy principle \cite{jaynes-1}, but it involves all 
the conserved charges of the model.

Even though the concept of the GGE is only  $\sim 10$ years
old, it has a quite rich history. Whereas it was fairly quickly proven
to be the correct thermodynamic ensemble in the case of free systems
\cite{mussardo-ising-1,mussardo-ising-2,ising-quench-1,ising-quench-2,ising-quench-3,ising-quench-4,ising-quench-5,essler-truncated-gge,gurarie-gge,second-quench,free-gge-1,free-gge-2,free-gge-3,free-gge-4},
the case of interacting lattice models had its twists and turns. After
some early to attempts to construct the GGE for the Heisenberg spin
chain \cite{sajat-xxz-gge,essler-xxz-gge} it was shown in
\cite{JS-oTBA,sajat-oTBA} that the GGE built on the known set of
strictly local charges fails to give correct predictions for the
steady states in particular quench situations. This failure was later
attributed to an {\it incompleteness} of the known charges, and the
work \cite{JS-CGGE} showed that a {\it Complete GGE} can be built by
incorporating the recently discovered quasi-local charges
\cite{prosen-xxx-quasi,prosen-enej-quasi-local-review} as
well.

After clarifying the GGE for the Heisenberg chains it became widely
accepted that there should be a 
complete GGE for any integrable model, and the remaining issue is 
to find the correct set of conserved charges. Whereas this might seem like
a relatively minor problem, it is far from being trivial in models
 more complicated than the XXZ spin chain. Ultimately one
would like a general proof for the existence of a Complete GGE, at
least after specifying the integrability structure of the
model. However, such a proof is not yet in sight.

The theory of Generalized Hydrodynamics (GHD) also motivates the
further study of the GGE. GHD describes large scale transport
properties of integrable models \cite{doyon-GHD,jacopo-ghd}, and one
of the main assumptions of the theory is the existence of local (space
and time dependent) GGE's, for which there exists a complete set of
charges. The GHD has been applied already for more complicated systems such
as the Hubbard model \cite{jacopo-enej-hubbard}, where this
completeness has not yet been studied. Therefore it is important to
understand the GGE in these more complicated models.

Here we contribute to the subject by formulating the GGE for a
prototypical multi-component system, namely the $SU(3)$-symmetric
fundamental spin chain, also known as the Lai-Sutherland
model. Furthermore, we present some conjectures about the GGE in the
fundamental $SU(N)$-symmetric model for any $N\ge 3$.
These models can be solved by the so-called nested Bethe Ansatz
\cite{kulish1981generalized,johannesson1986structure,johannesson-su3-sun,Hubbard-book}.
The eigenstates can be characterized by multiple sets of rapidities: the
first set describes the lattice momenta of the quasi-particles, which
are the excitations above a reference state, whereas the remaining sets
describes the wave function amplitudes in the internal space of the
spin waves.

We should note that even though the GGE for these
particular models has not yet been set up, specific quantum quenches in
the $SU(3)$ case have already been studied. They all involve so-called
integrable initial states \cite{sajat-integrable-quenches}, which
allow an exact analytic solution due to certain relations to boundary
integrability \cite{sajat-mps}. A quantum quench from a specific
Matrix Product State (MPS) was already studied in
\cite{nested-quench-1}, and local two-site states were investigated in
\cite{sajat-su3-1,sajat-su3-2}. The light cone spreading of
entanglement and correlations was studied in
\cite{pasquale-nested-entanglement}. Nevertheless the question of the
existence of a complete GGE in the $SU(3)$-symmetric model remained
completely open up to now. We should also note that quantum quenches in the $SU(N)$-symmetric
Fermi-Hubbard model have been studied recently in
\cite{fermi-hubbard-sun}, but here the theoretical framework only
dealt with a specific limiting case, and not the general Bethe Ansatz solution.

In the following subsection we give a more detailed description of the
GGE in generic integrable models and explain the main mechanisms
responsible for the emergence of the GGE, while omitting many
technical details.
Afterwards, the remainder of the paper is organized as follows:
In \cref{sec:SUNintro} we define the $SU(N)$ symmetric spin chain, and
consider the most basic properties of it. In \cref{sec:XXX} we discuss
the GGE in
the $N=2$ case, which is the celebrated XXX Heisenberg spin chain.
This Section includes known results, but we re-derive them using
slightly different techniques, more adequate for later
generalizations.
\Cref{sec:XXX} thus sets the stage for our investigations of
the multi-component models.
In \cref{sec:SU3} we
introduce the main interest of our paper, the $SU(3)$ symmetric
model, and discuss its main properties.
In
\cref{sec:DefAndConjRep} we consider two generating functions for conserved
charges, which correspond to the defining and conjugate
representations of $SU(3)$. We rigorously prove their quasi-locality. In
\cref{sec:SCrelations} we build two families of charges, and
we derive the complete set of 
string-charge
relations for this model. In \cref{sec:SUNGGE} we conjecture the
structure of the complete GGE for generic $SU(N)$.  We conclude in
\ref{sec:conclusions}, and a number of technical computations are
collected in the appendices.

\subsection{Generalized Eigenstate Thermalization}

Let us consider an integrable lattice model with a local Hamiltonian
$H$, defined in some finite volume $L$.

We consider a quantum quench situation, when the system is prepared in
the initial state $\ket{\Psi(t=0)}=\ket{\Psi_0}$.
We will consider different volumes and eventually the thermodynamic
limit. Therefore we require, that $\ket{\Psi_0}$ should be well-defined in
any finite volume and also as $L\to\infty$.
One possibility is to define $\ket{\Psi_0}$
as the ground state of a different local Hamiltonian,
or as a Matrix Product State (MPS) \cite{mps-intro1}.
We also require that $\ket{\Psi_0}$ satisfies the
cluster 
decomposition principle, namely for two local operators $\ordo_{1,2}(x)$
\begin{equation}
\lim_{|x_1-x_2|\to\infty}
\bra{\Psi_0}\ordo_1(x_1)\ordo_2(x_2)\ket{\Psi_0}=
\bra{\Psi_0}\ordo_1(x_1)\ket{\Psi_0}\bra{\Psi_0}\ordo_2(x_2)\ket{\Psi_0}.
\end{equation}

The state of the system at later times is given by
\begin{equation}
  \ket{\Psi(t)}=e^{-iHt}\ket{\Psi_0}.
\end{equation}

We are interested in equilibration and thermalization, to this order
we investigate the long time limit of local observables. Here and in
the following we understand that the $L\to\infty$ is taken {\it
  before} the long time limit. However, certain formal expressions are
more easily handled by keeping $L$ finite.

Before turning to the integrable models, let us focus on the more
simple non-integrable case. In generic non-integrable systems equilibration and thermalization to a Gibbs
Ensemble can be argued as follows \cite{rigol-eth}.

A direct finite volume expansion for the time evolution gives
\begin{equation}
 \vev{ \ordo(t)}=\sum_{j,k} \skalarszorzat{\Psi_0}{\Psi_j}\bra{\Psi_j}\ordo\ket{\Psi_k}
\skalarszorzat{\Psi_k}{\Psi_0}e^{-i(E_k-E_j)t}.
\end{equation}

In the long time limit dephasing leads to the emergence of the
Diagonal Ensemble:
\begin{equation}
  \lim_{T\to\infty}\left[ \int_0^T dt \vev{ \ordo(t)}\right]
    =\sum_{j} |\skalarszorzat{\Psi_0}{\Psi_j}|^2\bra{\Psi_j}\ordo\ket{\Psi_j}.
  \end{equation}
The Eigenstate Thermalization Hypothesis (ETH) states that for
almost all states in a small energy window $[E,E+\Delta E]$ the
mean values $\bra{\Psi_j}\ordo\ket{\Psi_j}$ will be close to each
other \cite{eth1,eth2}. Due to energy conservation the system will be populated only
with states that are close to each other in energy density, therefore
the diagonal ensemble has to be equal to the microcanonical
average. In large volumes the microcanonical and canonical
averages become equivalent for local operators, thus we have argued for the emergence of
the Gibbs Ensemble:
\begin{equation}
  \lim_{T\to\infty}\left[ \int_0^T dt \vev{ \ordo(t)}\right]=\vev{\ordo}_{\text{GE}}\equiv
  \frac{\text{Tr}\left(e^{-\beta H}\ordo\right)}{\text{Tr}\left(e^{-\beta H}\right)}.
\end{equation}
Here the parameter $\beta$ has to be chosen such that energy
conservation holds:
\begin{equation}
  \bra{\Psi_0}H\ket{\Psi_0}=\vev{H}_{\text{GE}}.
\end{equation}

In integrable models the situation is different due to the existence
of a large family of additional conserved charges. It has been known
since the early days of integrability that the integrable lattice model
possesses a family of commuting operators
\begin{equation}
  [\mathcal{Q}_j,\mathcal{Q}_k]=0,\qquad j,k=1,\dots,\infty,
\end{equation}
such that the Hamiltonian is a member of the series and each
$\mathcal{Q}_k$ is an extensive operator whose operator density is
strictly local. Typically it is possible to choose the charges such
that the density of $\mathcal{Q}_k$ spans $k$ sites.

The existence of these charges leads to the concept of the Generalized
Gibbs Ensemble (GGE). The main idea is to involve all conservation
laws in the standard statistical physical derivations. Based on the maximum
entropy principle we expect that the equilibrated values
of local observables will be given by
\begin{equation}
  \label{GGE}
\vev{\ordo}_{\text{GGE}}\equiv
  \frac{\text{Tr}\left(e^{-\sum_j \beta_j
        \mathcal{Q}_j}\ordo\right)}{\text{Tr}\left(e^{-\sum_j \beta_j\mathcal{Q}_j}\right)}.
\end{equation}
Here the generalized inverse temperatures $\beta_j$ are determined by the
initial state through 
\begin{equation}
  \bra{\Psi_0}\mathcal{Q}_j\ket{\Psi_0}=\vev{\mathcal{Q}_j}_{\text{GGE}},\qquad
  j=1,\dots,\infty,
\end{equation}
which are a set of coupled non-linear equations.

In analogy with the non-integrable case, where the ETH is the main
mechanism for the emergence of the GE, in integrable models
equilibration to the GGE is guaranteed if the Generalized Eigenstate
Thermalization (GETH) holds with the given set of conserved charges \cite{rigol-geth}. In
rough terms the GETH states that in the TDL the mean values of local observables
only depend on the global mean values of the conserved charges, and
not on any other details of the state.

It was realized in the case of the Heisenberg spin chains that the
GETH does not hold if we only consider the traditional set of local
charges \cite{JS-oTBA,sajat-oTBA,andrei-gge,sajat-GETH}. Instead, it
was realized that the so-called quasi-local charges
\cite{prosen-xxx-quasi,prosen-enej-quasi-local-review} need to be
included as well \cite{JS-CGGE,jacopo-massless-1}. The main reason for
this is the following.

In integrable models solvable by the Bethe Ansatz the finite volume eigenstates are
characterized by a finite set of Bethe rapidities. In the
thermodynamic limit the equilibrium
configurations are described by root distribution functions
$\rho_{\alpha}(\lambda)$, where $\lambda$ is the rapidity parameter
and $\alpha$ is an index or multi-index describing particle types. It
is a general understanding that in such models the local correlation
functions depend on all root densities. This was already postulated in early works \cite{Korepin-Book},
and it was proven for the XXZ
chain in \cite{sajat-corr,sajat-corr2}. According to this picture, a
set of conserved operators  is complete, {\it if their mean values completely fix all
the Bethe root densities}. This is the ultimate form
of the GETH, relevant for interacting integrable models. This idea was
further formalized in \cite{enej-gge,enej-gge-qtm}, where it was
argued that the GGE should be formulated
using root density operators, whose eigenvalues are the root densities
themselves. 

In Section \ref{sec:XXX} we summarize the known results of the XXX
chain and show that a complete set of quasi-local charges
indeed fixes all the root densities. It is the goal of our paper to extend
this picture to the $SU(N)$-symmetric chains with $N\ge 3$, and to find a complete
set of quasi-local operators.

To conclude this Section, we mention an important property of the
GGE construction, which is inherently linked to the
question, whether the GGE can be regarded as a statistical physical ensemble.
 The complete GGE does not use the 
maximum entropy principle, because it fixes all Bethe root densities.
The entropy principle is used in a different approach, the
``Quench Action'' method \cite{quench-action}, which is applicable
only for so-called integrable quenches
\cite{sajat-integrable-quenches,sajat-minden-overlaps}.
Therefore it is our opinion that
the term ``GGE'' is not adequate for quenches in integrable models, and we should instead use
``Generalized Eigenstate Thermalization'' or similar, alternative
names. The term ``GGE'' is useful due to historical reasons, but
it does not capture the essence of equilibration in integrable models.

\section{The $SU(N)$-symmetric spin chains - generalities}

\label{sec:SUNintro}
Let us consider an integer $N\ge 2$. We define a spin chain with local
Hilbert spaces $h_j = \mathbb{C}^N$ called quantum spaces such that the full Hilbert space
of the chain if length $L$ is $\mathcal{H}_L= h_1 \otimes
h_2 \otimes \dots \otimes h_L = \big(\mathbb{C}^N\big)^L$.

We consider the fundamental $SU(N)$-symmetric model \cite{lai1974lattice,
  sutherland1975model}  defined on this Hilbert space, given by the Hamiltonian
\begin{equation}
\label{eq:SU3Hamiltonian}
 H = -L + \sum_{j=1}^L P_{j,j+1}.
\end{equation}
Above $P\in\text{End}(\mathbb{C}^N\otimes\mathbb{C}^N)$ is the
permutation operator, which acts as $P(v_1\otimes v_2) = v_2\otimes
v_1, \, v_1, v_2\in\mathbb C^N$. For simplicity we consider the model
under periodic boundary conditions: $P_{L,L+1}=P_{L,1}$.

For $N=2$ the model is equivalent to the famous Heisenberg XXX spin
chain, whereas for $N\ge 3$ it can be considered a higher rank
generalization of it. 

One of the most important properties of the Hamiltonian
(\ref{eq:SU3Hamiltonian}) is its $SU(N)$ invariance, which is
understood as follows. Let the local Hilbert spaces $h_j$ carry the
defining representation of $SU(N)$, and let us extend the group action
to the tensor product. Then the global $SU(N)$ invariance of the
Hamiltonian immediately follows from the fact that it involves 
invariant local operators.

The model is integrable for any $N$: it possesses an infinite family of commuting
local charges, and it can be solved by the Algebraic Bethe Ansatz.
The exact real space wave functions of the eigenstates are given by the so-called 
nested Bethe Ansatz
\cite{kulish1981generalized,johannesson1986structure,johannesson-su3-sun,Hubbard-book}.
In the following we briefly review the standard integrability
framework of these models.
We focus on the construction of the commuting set of transfer matrices, and their
 eigenvalues expressed in terms of Bethe Ansatz rapidities. We do not
 treat the actual construction of the Bethe states, and we refer the
 reader to \cite{resh-su3,kulish-resh-glN}.

Let us consider the following fundamental $R$-matrix:
\begin{equation}
  \label{Rdef}
R(u) = \frac{1}{u+\ii}\left(u+\ii P\right), \qquad R\in\text{End}\big(\complex^N\otimes\complex^N\big),
\end{equation}
which satisfies the Yang-Baxter equation \cite{Baxter-Book}
\begin{align}
  \label{YB}
 R_{23}(v-w)R_{13}(u-w)R_{12}(u-v) &= R_{12}(u-v)R_{13}(u-w)R_{23}(v-w),
\end{align}
and the unitarity relation:
\begin{equation}
  \label{eq:uni1}
   R(u) R(-u) = 1.
\end{equation}
It is group invariant with respect to $GL(N)$:
\begin{equation}
  \label{groupinvariance}
 G_1 G_2 R(u) = R(u) G_1 G_2, \quad G_i \in GL(N), i=1,2.
\end{equation}

Let us consider an additional space $h_0= \mathbb{C}^N$ called the
auxiliary space. We define the transfer matrix (TM) of the model in the usual way:
\begin{equation}
  \label{tdef}
 t(u) = \Tr_0 R_{10}(u)R_{20}(u)\dots R_{L0}(u),
\end{equation}
where the trace is taken on the auxiliary space $h_0$. The transfer
matrices form a  commuting family:
\begin{equation}
  \label{tt}
 \lbrack t(u), t(v) \rbrack = 0.
\end{equation}
The commuting set of local charges is built from $t(u)$. Let
\begin{equation}
  \label{Qdef}
 \mathcal{Q}_k = (-\ii)\left.\frac{d^{k-1}}{du^{k-1}} \log
   t(u)\right|_{u=0},\qquad k\ge 2.
\end{equation}
It can be shown that the $ \mathcal{Q}_k$ are local charges: they are extensive
such that their operator density spans at most $k$ sites \cite{Luscher-conserved}. It follows
from \eqref{tt} that they commute with each other:
\begin{equation}
  [ \mathcal{Q}_j, \mathcal{Q}_k]=0,\qquad j,k\ge 2.
\end{equation}
Furthermore, the Hamiltonian is a member of this series. Direct
computation gives $H=- \mathcal{Q}_2 $.

In our work a special role will be played by the so-called fusion
hierarchy of the transfer matrices. In the following we briefly
introduce these concepts, while omitting many  technical details.

Let $\Lambda_1$ and $\Lambda_2$ be two irreducible
representations of $SU(N)$. It is known that there exists an
$R$-matrix $R^{\Lambda_1,\Lambda_2}(u)$ acting on the tensor product
of the two representations, which is unique up to an overall scaling
and certain shifts in the rapidity parameter, such that for any three
representations $\Lambda_j,\ j=1,2,3$ they satisfy the Yang-Baxter
equation \cite{kulish-resh-sklyanin--fusion}:
\begin{equation}
\label{eq:genYBE}
 R_{23}^{\Lambda_2,\Lambda_3}(v-z)R_{13}^{\Lambda_1,\Lambda_3}(u-z)R_{12}^{\Lambda_1,\Lambda_2}(u-v)
 = R_{12}^{\Lambda_1,\Lambda_2}(u-v)R_{13}^{\Lambda_1,\Lambda_3}(u-z)R_{23}^{\Lambda_2,\Lambda_3}(v-z).
\end{equation}
These $R$-matrices can be obtained by the so-called fusion procedure
\cite{kulish-resh-sklyanin--fusion,junji-suzuki-kuniba-tomoki-fusion}.

In our models  each local spin variable carries the defining
representation of $SU(N)$, therefore we will need $R$-matrices acting
on the tensor product of the defining representation and some other $\Lambda$. 
For these cases we will use the 
short notation $R^{\Lambda}_{12}(u)$.

For each representation $\Lambda$ we  define the transfer matrix
with auxiliary space carrying $\Lambda$ as
\begin{equation}
\label{eq:GenericTmx}
 t^{\Lambda} (u) = \Tr_0 R_{10}^\Lambda (u)R_{20}^\Lambda (u)\dots R_{L0}^\Lambda (u).
\end{equation}

It follows from \eqref{eq:genYBE} that all of these transfer matrices
commute:
\begin{equation}
  [t^\Lambda(u),t^{\Lambda'}(v)]=0.
\end{equation}

The representations of $SU(N)$ can be described by Young diagrams. The
transfer matrices corresponding to rectangular Young diagrams play a
special role in the theory, and we will show that they are central
also for the GETH. For the Young diagram with $a$ rows and $s$
columns the corresponding transfer matrix will be denoted as $t^{(a)}_s(u)$. These
objects satisfy a closed set of functional relations called the Hirota
equation or $T$-system relations; specific details will be given
later, and for reviews see
\cite{zabrodin-hirota-1,suzuki-kuniba-tomoki-t-system-y-system-review}.

The common eigenstates of the transfer matrices can be found by the
(nested) Bethe Ansatz. The actual construction, and hence the discussion of
the GGE and the GETH strongly depends on $N$.
In the next section we review the known results for the case
$N=2$, which corresponds to the XXX Heisenberg spin chain. In
\ref{sec:SU3} we start our discussion of the $SU(3)$-symmetric chain,
which is the main subject of this paper. However, before going to
these special cases we introduce the notion of quasi-local charges,
that are essential for the GETH.

\subsection{Quasi-local charges}

The set of canonical local charges \eqref{Qdef} has been known since
the early days of integrability. On the other hand, the existence and
importance of quasi-local charges was only understood in recent years
\cite{prosen-xxx-quasi,prosen-enej-quasi-local-review}. Here we define
quasi-locality following  the works \cite{prosen-xxx-quasi,prosen-enej-quasi-local-review}.

Consider the physical Hilbert space $\mathcal H_L =
h_1\otimes\dots\otimes h_L = \left(\mathbb C^N\right)^L$, and consider
$\text{End}(\mathcal H_L)$, the space of linear operators over
$\mathcal H_L$. $\text{End}\mathcal (H_L)$ possesses a Hilbert space
structure, under the Hilbert-Schmidt scalar
product, defined as  
\begin{equation}
 \langle A,B \rangle_{\text{HS}} = N^{-L}\Tr \left(A^\dagger B\right), \quad A, B \in\text{End}\left(\mathcal H_L\right).
\end{equation}
The Hilbert-Schmidt norm is  defined as
\begin{equation}
 \| A \|_{\text{HS}}^2 = \langle A,A \rangle_{\text{HS}} = N^{-L} \Tr \left( A^\dagger A \right).
\end{equation}
This normalization is such, that for the identity operator $||1||_{\text{HS}}=1$.

We define the traceless part of an operator $A$ as 
\begin{equation}
 \{ A \} = A - N^{-L} \Tr (A).
\end{equation}
Quasi-locality is defined for traceless operators.

The $L$-dependent operator $\{A\}(L)\in\mathcal H_L$ is called
quasi-local, if it satisfies the following properties: 
\begin{enumerate}
  \item $\{A\}(L)$ is translationally invariant for every $L$. 
 \item In large volumes $\| \{A\} \|_{\text{HS}}^2 \sim L $
 \item For any locally supported $k$-site operator $b= b_k \otimes
   1_{L-k}$ the overlap $\langle b,\{A\} \rangle_{\text{HS}}$
   is asymptotically independent of $L$ in the $L\rightarrow\infty$.
   \end{enumerate}
Note that the quasi-locality only makes sense in the strictly
$L\rightarrow\infty$ limit, because it is the property of the infinite
series of operators $\{A\}(L)$. The strictly local charges obviously
satisfy these requirements.

A quasi-local operator can be written in the form
\begin{equation}
  A(L)=\sum_{x=1}^L a(x),
\end{equation}
where $a(x)$ is the operator density of $A$. It does not have to be
local, but it has to have a finite norm. As an effect, the long-range
contributions to $a(x)$ have (typically exponentially) decreasing amplitudes.

\section{The Heisenberg spin chain}

\label{sec:XXX}

The $SU(2)$-symmetric Heisenberg XXX spin chain is defined conventionally by the Hamiltonian
\begin{equation}
  H_{XXX}=\sum_{j=1}^L \left(
    \sigma_{j}^x\sigma_{j+1}^x+ \sigma_{j}^y\sigma_{j+1}^y+ \sigma_{j}^z\sigma_{j+1}^z-1,
\right),
  \end{equation}
where $\sigma_j^\alpha$ are the Pauli matrices acting on quantum space $j$. 
In this normalization $H_{XXX}=2H$, where $H$ is the Hamiltonian
\eqref{eq:SU3Hamiltonian} at $N=2$.

The exact eigenstates of this model were first found by Bethe
\cite{Bethe-XXX}. They can be characterized by a set of rapidities
$\{\lambda_1,\dots,\lambda_N\}$, which parametrize the quasi-momenta
of interacting spin waves above the reference state $\ket{\emptyset}$, which is chosen
conventionally as the state with all spins up. The un-normalized
eigenstates can be written as
 \begin{equation} 
  \begin{split}
| \boldsymbol \lambda_N\rangle &= \sum_{1\leq n_1<\ldots <n_N\leq L}\ 
  \Psi(n_1,\dots,n_N) \prod_{s=1}^{N}(\sigma_-)_{n_s}\ket{\emptyset}\\
\Psi(n_1,\dots,n_N)&= \sum_{\mathcal P \in \mathcal S^N} \left(\prod_{1\leq r < l \leq N}\frac{\lambda_{\mathcal P (l)}-\lambda_{\mathcal P (r)}-\ii}{\lambda_{\mathcal P (l)}-\lambda_{\mathcal P (r)}}
\prod_{r=1}^N\left(\frac{\lambda_{\mathcal{P}(r)}+\frac\ii2}{\lambda_{\mathcal{P}(r)}-\frac\ii2}\right)^{n_r}\right).
\label{eq:bethe_stateXXX}    
  \end{split}
\end{equation}
Here $\sigma_-$ is the spin lowering operator, and the $n_s$ describe
the positions of the spin waves.

These states are eigenstates of the Hamiltonian if the wave functions
are periodic, which leads to the Bethe equations:
\begin{equation}
    \begin{split}
 \left(\frac{\lambda_j+\frac{\ii}{2}}{\lambda_j-\frac{\ii}{2}}\right)^L &= \prod_{k=1,k\ne j}^N \frac{\lambda_j - \lambda_k +\ii}{\lambda_j - \lambda_k -\ii} \qquad j=1,\dots,N \label{eq:BEXXX}.
    \end{split}
  \end{equation}
The lattice momentum and the energy of a
Bethe state $\ket{\{\lambda_j\}_{j=1}^N}$ are given
by the sum of one particle momentum and energy, respectively: 
\begin{align}
  \label{PE0}
 P &=\sum_{j=1}^N p(\lambda_j), & p(\lambda)& = \ii \log \left(\frac{\lambda_j+\frac{\ii}{2}}{\lambda_j-\frac{\ii}{2}}\right), \\
 E &=\sum_{j=1}^N \varepsilon(\lambda_j), & \varepsilon (\lambda)& =- \frac{1}{\lambda^2+\frac1 4}. 
\end{align}
The eigenvalues of the fundamental transfer matrix \eqref{tdef} can be
computed using Algebraic Bethe Ansatz \cite{Korepin-Book}.  On the
Bethe state $\ket{\{\lambda\}_N}$ the eigenvalues are
\begin{align}
  \label{tu0}
 t(u)=
        \frac{Q_1\left(u-\frac\ii2\right)}{Q_1\left(u+\frac\ii2\right)}+
        \frac{Q_0(u) Q_1\left(u+\frac{3\ii}2\right)}{Q_0(u+\ii) Q_1\left(u+\frac\ii2\right)},
   \end{align}
where
\begin{equation}
 Q_0 (u) = u^L,\qquad
 Q_1 (u) = \prod_{j=1}^N (u-\lambda_j).
\end{equation}
The Bethe states are highest weight with respect to the global $SU(2)$
symmetry; other states in the same multiplet can be obtained by global
spin lowering operators.

\subsection{String hypothesis and Thermodynamic Bethe Ansatz}

In order to study the thermodynamic limit of the model it is important
to know the positions of the Bethe roots in the complex plain. It is
known that in Bethe Ansatz solvable models the roots typically arrange themselves into so-called string
patterns. A string describes a bound state of spin waves. The string
hypothesis states that in the thermodynamic limit the dynamical
processes can be described by concentrating only on the regular string
solutions, neglecting contributions from the rare outlier states.

In the XXX model an $n$-string pattern 
 centered around the rapidity $x\in\valos$ takes the form
\begin{equation}
  x^\ell=x+\ii\left( \frac{n+1}2 -   \ell \right) +\delta^{\ell} , \quad \ell=1,\dots,n.
\end{equation}
Here $\delta^\ell$ are the so-called string deviations which are
exponentially small in large volumes for regular string solutions.

A Bethe state with $N_n$ number of $n$-strings
thus consists of the rapidities
\begin{align}
 \lambda_\alpha^{n,\ell} &=\lambda_\alpha^n +\ii\left( \frac{n+1}2 -   \ell \right) +\delta_{\alpha}^{n,\ell}
                           , \quad \ell=1,\dots,n,\ \quad \alpha = 1, \dots, N_n.
\end{align}
The total number of Bethe roots is computed simply as
\begin{equation}
  \sum_{n=1}^\infty n N_n = N.
\end{equation}

We will be interested in the thermodynamic limit, when
$N\rightarrow\infty$, $L\rightarrow\infty$ with a fixed $N/L$ ratio.
In this limit the
string deviations become exponentially small and according to the string hypothesis
it is sufficient to describe the positions of the string centers.
Then the string centers can be described
by continuous density functions along the real line. These are denoted by
$\rho_n(u)$, and are normalized such that in a large
volume $L$ the total number of $n$-strings
between rapidities $u$ and $u+\Delta u$ is
$L\rho_n(u)\Delta u$.

Analogously to a free system we also introduce hole densities.
For a Bethe state a hole is a position in rapidity space which would
satisfy the Bethe equations but it is not actually a Bethe root. Such
holes can be defined for each particle type and each string pattern,
and in the TDL they are described by the densities $\rho_{h,n}$.

In the thermodynamic limit the Bethe
equations  \eqref{eq:BE} can be transformed into a set of coupled linear integral equations: 
\begin{align}
  \label{rootXXX}
 \rho_{t,n} (u)&= \delta_{n,1}s(u) + s\star\left(\rho_{h,n-1}^{(1)} + \rho_{h,n+1}^{(1)} \right) (u),
\end{align}
where $\rho_{t,n}^{(j)},\ j=1,2$ are the so-called total root densities, defined by
\begin{equation}
 \rho_{t,n}^{(j)}(u) = \rho_{n}^{(j)}(u)+\rho_{h,n}^{(j)}(u),\qquad j=1,2.
\end{equation}
Furthermore, the convolution is understood as
\begin{equation}
 (f\star g)(\lambda) = \int_{-\infty}^\infty d\mu\ f(\lambda-\mu)g(\mu),
\end{equation}
and the integration kernel is
\begin{equation}
  \label{sdef}
  s(\lambda) = \frac{1}{2\cosh(\pi\lambda)}.
\end{equation} 
The integral equations \eqref{rootXXX} are symbolically depicted on
Fig. \ref{fig:xxxtba}.

\begin{figure}
  \centering
  \begin{tikzpicture}
      \draw[thick,fill] (0,0) circle (0.1);   
    \foreach \i in {0,1,2,3,4}
          {
      \begin{scope}[xshift=\i cm]
      \draw[thick] (0,0) circle (0.1);       
\draw[thick] (0.1,0) to (0.9,0);
    \end{scope}
  }
  \node at (5.5,0) {$\dots$};
  \end{tikzpicture}
  \caption{TBA diagram of the XXX model. Each node corresponds to a
    string type, and the links denote the convolution kernels in
    Eq. \eqref{rootXXX}. The filling of the first node signals the
    source term for $\rho_{t,1}(u)$.}
  \label{fig:xxxtba}
\end{figure}
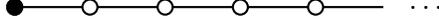

\subsection{String-charge relations}

The main question of the GETH is to what extent a given set of charges
determines the Bethe root densities. Let us first focus on the set of
strictly local charges $\{\mathcal{Q}_\alpha\}$ defined in \eqref{Qdef}. Instead
of dealing with the discrete set it is useful to consider the
generating function $X_1(u)$ defined formally as
\begin{equation}
  X_1(u)= (-\ii)\partial_u \log t (u)=\sum_{j=2}^\infty \frac{u^{j-2}}{(j-2)!}\mathcal{Q}_j.
\end{equation}
It was shown in \cite{essler-xxz-gge,JS-oTBA} that the eigenvalues of
this operator are asymptotically
\begin{equation}
  \frac{1}{2\pi L}X_1(u)=s\star (\rho_{h,1}+a_1),
\end{equation}
where 
\begin{equation}
  a_1(u)=\frac{1}{2\pi}\frac{1}{u^2+\frac{1}{4}}.
\end{equation}
It follows that this set of charges is not sufficient to determine all
Bethe root densities: $X_1(u)$ only fixes the hole density of the $1$-strings.

This situation was remedied in \cite{JS-CGGE}
(see also \cite{jacopo-massless-1,enej-gge}), where it was shown that
the recently introduced quasi-local charges 
\cite{prosen-xxx-quasi,prosen-enej-quasi-local-review}
contain just enough information to fix all the root densities.

The quasi-local charges are obtained from the fusion hierarchy of the
transfer matrices.
Let us define the higher spin Lax operators with spin $s=m/2,\ m\in\mathbb{N}$
acting on the tensor product $V_a\otimes V_j=\complex^{2s+1}\otimes \complex^2$ as
\begin{equation}
  \label{su2lax}
  \mathcal{L}^m_{a,j}(u)=\frac{u+\ii \frac{1}{2}+\ii\boldsymbol S_a \cdot \boldsymbol \sigma_j}{u+\ii\frac{m+1}{2}},
\end{equation}
where $\boldsymbol S_a$ stands for the vector of the spin-$s$
generators of $SU(2)$, and $\boldsymbol \sigma_j$ is the vector
constructed out of Pauli matrices.
Our conventions for the Lax operators differs slightly from the one
used in \cite{prosen-xxx-quasi,JS-CGGE,jacopo-massless-1}.

We define the corresponding transfer matrices (TMs)
\begin{equation}
  \label{highertdef}
 t_m(u) = \Tr_a   \mathcal{L}^m_{a,1}(u)  \mathcal{L}^m_{a,2}(u)\dots   \mathcal{L}^m_{a,L}(u).
\end{equation}
It can be shown that these operators form a commuting family:
\begin{equation}
  [t_m(u),t_{n}(v)]=0.
\end{equation}
The spin-$s$ representations correspond to Young diagrams with 1 row
and $m$ columns, therefore using the notations of Section
\ref{sec:SUNintro} we have the identification $t_m(u)=t^{(1)}_m(u)$.

These transfer matrices satisfy a set of
functional equations called the Hirota equation or $T$-system:
\begin{equation}
  \label{su2t}
  t_m\left(u+\frac\ii2\right)t_m\left(u-\frac\ii2\right)=t_{m+1}(u)t_{m-1}(u)+\phi_m(u).
\end{equation}
Here $\phi_m(u)$ is a scalar function (independent of the Bethe state)
given by
\begin{equation}
  \phi_m(u)=\frac{Q_0\left(u-\ii\frac{m}{2}\right)}{Q_0\left(u+\ii\frac{im}{2}\right)}.
\end{equation}
Furthermore, the initial value for the recursion is $t_0(u)=1$.

The eigenvalues on the common eigenstates are \cite{bazhanov2007analytic}
\begin{equation}
  \label{tmueig}
  t_m(u)=\frac{Q_1\left(u-\ii\frac m2\right)Q_1\left(u+\ii\frac {m+2}2\right)}{Q_0\left(u+\ii\frac {m+1}2\right)} \sum_{k=0}^m \frac{Q_0\left(u+\ii\frac {m+1}2-\ii k \right)}{Q_1\left(u+\ii\frac {m}2 -\ii k\right)Q_1\left(u+\ii\frac {m+2}2 -\ii k\right)}.
\end{equation}
It can be checked by direct computation that these eigenvalues satisfy
the $T$-system \eqref{su2t}.

A key role is played by the operators $X_m(u)$ defined formally by
\begin{equation}
  \label{xmdef}
   X_m (u) = (-\ii)\partial_u \log t_m (u).
\end{equation}
It was shown in \cite{prosen-xxx-quasi} that the traceless operators
$\{X_m(u)\}$ are quasi-local in the thermodynamic limit, if $u$ is
within the physical strip $\mathcal{P}$ defined as
\begin{equation}
  \mathcal{P}\equiv \left\{u\in\complex,\quad |\Im(u)|<\frac{1}{2}\right\}.
\end{equation}
For a more precise treatment of $X_m(u)$ see the next subsection.

Regarding the eigenvalues of the operators  $X_m(u)$ it was obtained in \cite{JS-CGGE}
\begin{equation}
  \label{su2sc0}
\rho_{h,n}=a_n- \frac{1}{2\pi L}(X_n^{[+]}+X_n^{[-]}),
\end{equation}
where
\begin{equation}
  \label{andef}
   a_n (\lambda) = \frac{1}{2\pi}\frac{n}{\lambda^2+\frac{n^2}{4}},
\end{equation}
and we introduced the short-hand notation:
\begin{equation}
\begin{split}
 f^{[\pm]}(u) &= \lim_{\eps\to 0} f\left(u\pm\frac{\ii}{2}\mp\eps\right)
\end{split}
\end{equation}
Making use of the system \eqref{rootXXX} an equivalent form can be derived:
\begin{equation}
\label{su2sc}
     \rho_{n}= \frac{1}{2\pi L}\left(X_n^{[+]}+X_n^{[-]}-X_{n-1}-X_{n+1}\right).
\end{equation}
Thus the higher spin transfer matrices contain just enough information to determine
all the root densities.

Let us comment on some important differences between the finite volume
situation and the thermodynamic limit. 

In finite volume it is known that the spectrum of the transfer matrix
is simple \cite{tarasov-varchenko-xxx-simple}. This means 
that if two states possess the same eigenvalue function $t(u)$ then
they belong to the same $SU(2)$ multiplet. This also implies that if
the spin quantum number $S_z$ is also specified, then the function
$t(u)$ uniquely determines all Bethe roots. A practical procedure for
recovering the Bethe roots from $t(u)$ is explained for example in
\cite{XXXmegoldasok}.

Based on this, it might seem surprising, that the complete family
$\{t_m(u)\}$ of TM's is needed in the $L\to\infty$ limit.
Eq. \eqref{su2t} shows that the higher spin transfer matrices are
algebraically dependent, and they can be expressed using the
fundamental $t(u)$, thus the information stored in the complete family $\{t_m(u)\}$ might
seem redundant.

The explanation for this apparent paradox is the following. 
Even though at
finite $L$ the function $t(u)$ is enough the recover all Bethe roots,
typically a large amount of information is lost by the thermodynamic
limit. On a technical level this happens because for almost all $u$ one of the
two terms in the expression \eqref{tu0}  becomes dominant, and the other one becomes
exponentially suppressed as $L\to\infty$. Thus it becomes impossible
to reconstruct the root densities once the thermodynamic limit has
been taken.
However, further information is preserved in the other
members of the family $\{t_m(u)\}$, such that eventually the set $\{X_m(u)\}$
remains complete in the TDL.

It is our goal to extend this picture to the higher rank
cases. We will show that the situation is analogous to the
$SU(2)$ case: the complete set of charges is obtained from the fusion
hierarchy of the transfer matrices. However, before turning to the
$SU(3)$ case 
we repeat some of the computations already present in the literature.
We will use a slightly different approach, which is more convenient for
later generalizations to the higher rank cases.

\subsection{Inversion and quasi-locality}

\label{sec:invqloc}

In our computations an important role will be played by certain
asymptotic inversion relations. The main goal is to find some
operators that invert the transfer matrices, such that the formal
expressions $\partial_u \log(t_m(u))=(t_m(u))^{-1}\partial_u t_m(u)$
can be made sense using well defined local objects. The transfer
matrices themselves can not be inverted in the desired way, but there exist
asymptotic inversion relations which hold in the $L\to\infty$
limit. 

Such inversion relations are closely tied to the fusion of transfer
matrices. Their study has a long history, which goes back to
the seminal work of Baxter \cite{Baxter-Book}. We do not attempt a
thorough review of this topic, we merely mention a few
references. For example, we will rely on some basic arguments about the
inversion that already appeared in the
 work \cite{pearce-inversion} of Pearce. Closely related ideas
and methods appeared among others in \cite{resh-fusion,kuniba-suzuki-fusion-2}.

All of the asymptotic inversions that we will treat are based on a
local inversion. In the case of the XXX model
the Lax operators \eqref{su2lax} satisfy the local inversion 
\begin{equation}
  \label{localinversion}
 \mathcal{L}^m(u)  \mathcal{L}^m(-u)=
  \frac{-u^2-( \boldsymbol S_a \cdot \boldsymbol \sigma_j+\frac{1}{2})^2}{-u^2-(s+\frac{1}{2})^2}
  =1.
\end{equation}
This is most easily seen using the relation
\begin{equation}
  \label{Sasj}
  \boldsymbol S_a \cdot \boldsymbol \sigma_j+\frac{1}{2}=
  \frac{(  \boldsymbol S_a + \boldsymbol \sigma_j)^2-(\boldsymbol
    S_a)^2-(\boldsymbol \sigma_j)^2+1}{2}.
\end{equation}
From the known values of the Casimir operators we compute the two possible
eigenvalues of the operator in \eqref{Sasj} as $\pm (s+1/2)$,
which implies the inversion \eqref{localinversion}.

Let us define a new family of transfer matrices, which are obtained
simply by space reflection:
\begin{equation}
  \label{highertdef2}
\bar t_m(u) = \Tr_a \left[  \mathcal{L}^m_{a,L}(u)  \mathcal{L}^m_{a,L-1}(u)\dots   \mathcal{L}^m_{a,1}(u)\right].
\end{equation}
Using partial transpose in auxiliary space they can be expressed as
\begin{equation}
  \label{highertdef3}
  \bar t_m(u) = \Tr_a\left[
  \left( \mathcal{L}^m_{a,1}(u) \right)^{t_a}
  \left( \mathcal{L}^m_{a,2}(u) \right)^{t_a}
  \dots\relax
  \left( \mathcal{L}^m_{a,L}(u) \right)^{t_a}
\right].
\end{equation}
Furthermore, they can be related to the standard transfer matrices by a simple
crossing transformation. All representations of $SU(2)$ are
self-conjugate, therefore there exists a charge conjugation operator
$C$  acting on the auxiliary space, such that it satisfies $C^2=1$ and
performs conjugation as
\begin{equation}
  \boldsymbol S^t= \boldsymbol S^*=- C\boldsymbol S  C.
\end{equation}
Therefore
\begin{equation}
    \left( \mathcal{L}^m_{a,1}(u)\right)^{t_a}=C_a
\frac{u+\ii\frac{1}{2}-\ii\boldsymbol S_a \cdot \boldsymbol \sigma_j}{u+\ii\frac{m+1}{2}}
C_a=
\frac{u-\ii\frac{m-1}{2}}{u+\ii\frac{m+1}{2}}C_a
 \mathcal{L}^m_{a,1}(-u-\ii)C_a.
\end{equation}
The $C$-operators drop out when we compute the transfer matrices, we
thus obtain the global crossing relation
\begin{equation}
  \label{tmcrossing}
 \bar t_m(u) =  \left(\frac{u-\ii\frac{m-1}{2}}{u+\ii\frac{m+1}{2}}\right)^L
t_m(-u-\ii).
\end{equation}
In the following we will use
the direct definition
\eqref{highertdef2}, because  it is more advantageous for our purposes.

\begin{thm}
  The following asymptotic inversion identity holds \cite{prosen-xxx-quasi}:
\begin{equation}
  \label{su2inv}
  \bar t_m(-u)t_m(u)\approx 1,\qquad u\in \mathcal{P}.
\end{equation}
Here and in the following an asymptotic identity $A\approx B$ means
that
\begin{equation}
||A-B||=\ordo(e^{-\alpha L}),\quad \alpha\in\valos^+.
\end{equation}
\end{thm}

The above Theorem was treated in detail in \cite{prosen-xxx-quasi}. We
also sketch the proof using our conventions. First we shed some light on why the
inversion relation holds.

Let us think about a different situation, and consider the monodromy matrices, i.e. the
expressions \eqref{highertdef}-\eqref{highertdef2} without the
trace in auxiliary space:
\begin{equation}
  \begin{split}
M^m(u)&= \mathcal{L}^m_{a,1}(u)  \mathcal{L}^m_{a,2}(u)\dots   \mathcal{L}^m_{a,L}(u)\\
    \bar M^m(u)&= \mathcal{L}^m_{a,L}(u)  \mathcal{L}^m_{a,L-1}(u)\dots   \mathcal{L}^m_{a,1}(u).
  \end{split}
\end{equation}
In this case, the local inversion identity \eqref{localinversion} immediately gives a
global and exact inversion:
\begin{equation}
   \bar M^m(-u) M^m(u)=1.
\end{equation}
The local steps leading to this global inversion are depicted in Fig. \ref{fig:localinversion}.

For the transfer matrices \eqref{su2inv} the difficulty lies in the fact that
the trace has been taken
in auxiliary space. In this case we can not expect an exact
inversion. Nevertheless, for certain values of $u$ and for large
enough volumes the ``boundary effect'' of taking the trace does not
propagate into the bulk of the chain. More precisely, it only
causes an exponentially small effect. 

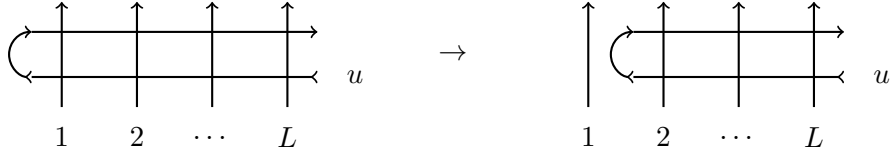
\begin{figure}
  \centering
  \begin{tikzpicture}
     \node at (4.3,0) {$u$};
    \draw [thick,-<] (0,0) to (3.8,0);
    \draw [thick,<-] (3.8,0.6) to (0,0.6);
    \draw [thick,>->] (0,0) [out=180,in=270] to (-0.3,0.3) [out=90,in=180]
    to (0,0.6);
    \foreach \i in {0,1,2,3}
    {
      \begin{scope}[xshift=\i cm]
      \draw [thick,->] (0.4,-0.4) to (0.4,1);      
      \end{scope}  
    }
    \node at (0.4,-0.8) {$1$};
    \node at (1.4,-0.8) {$2$};
        \node at (2.4,-0.8) {$\dots$};
    \node at (3.4,-0.8) {$L$};

\node at (5.6,0.3) {$\to$};
    
    \begin{scope}[xshift=7cm]
       \node at (4.3,0) {$u$};
         \draw [thick,-<] (1,0) to (3.8,0);
    \draw [thick,<-] (3.8,0.6) to (1,0.6);
    \draw [thick,>->] (1,0) [out=180,in=270] to (0.7,0.3) [out=90,in=180]
    to (1,0.6);
    \foreach \i in {0,1,2,3}
    {
      \begin{scope}[xshift=\i cm]
      \draw [thick,->] (0.4,-0.4) to (0.4,1);      
      \end{scope}  
    }
     \node at (0.4,-0.8) {$1$};
    \node at (1.4,-0.8) {$2$};
        \node at (2.4,-0.8) {$\dots$};
    \node at (3.4,-0.8) {$L$};
    \end{scope}
  \end{tikzpicture}
  \caption{Graphical representation of the global inversion, which
    follows from consecutive local inversion steps. The
    product of monodromy matrices is equal to the identity after
    disentangling the local Lax operations.}
  \label{fig:localinversion}
\end{figure}

A rigorous proof can be given by computing the norm of the difference
\begin{equation}
  \label{ezkell}
  \begin{split}
     || \bar t_m(-u)t_m(u)-1||^2=&
2^{-L}    \text{Tr}\left(( \bar t_m(-u)t_m(u))^\dagger \bar t_m(-u)t_m(u)\right)-\\
&-2^{-(L-1)}\Re\left[    \text{Tr}\left( \bar t_m(-u)t_m(u)\right)\right]
      +1.
        \end{split}
\end{equation}
The adjoints of the transfer matrices can be computed as
\begin{equation}
  \left(t_m(u)\right)^\dagger=\bar t_m(-u^*).
\end{equation}
Our aim is to show that
\begin{equation}
  \label{su224}
  \begin{split}
2^{-L}    \text{Tr}\left( \bar t_m(-u)t_m(u)\right)&\approx 1\\
    2^{-L}    \text{Tr}\left(t_m(u^*) \bar t_m(-u^*) \bar t_m(-u)t_m(u)\right)&\approx 1,
  \end{split}
\end{equation}
which will imply the asymptotic inversion. 
These two traces can be evaluated conveniently by building a
corresponding 2D vertex model, see Figs. \ref{fig:localinversion}-
\ref{fig:localinversion2}. Here the action of the transfer matrices
corresponds to adding a new row to the lattice, therefore
we get two lattices of size $2\times L$ and $4\times L$. These
partition functions can be evaluated in the ``crossed channel'' by
building column-to-column transfer matrices. These are conventionally
called Quantum Transfer Matrices (QTM's). The traces are evaluated using the eigenvalues of these
QTM's. It follows from the local inversion relations, that the local
delta-states given by $\ket{\delta}=\sum_{j} \ket{j}\otimes\ket{j}$ and
$\ket{\delta}\otimes\ket{\delta}$ are eigenstates of the two-site and
four-site QTM's, respectively. Their eigenvalues are simply 2, due to
the local inversion and the trace over the physical space,  see again
Figs. \ref{fig:localinversion}-\ref{fig:localinversion2}. The
identities in \eqref{su224} are rigorously proven by showing that the
delta-states are the dominant eigenstates. This requires a
diagonalization of the QTM's in question. In the case of the two-site
QTM this was performed in
\cite{prosen-xxx-quasi} using the $SU(2)$ algebra, whereas for
the four-site case it was done analytically up to $s=3/2$ and
numerically for larger values of $s$. 

We put forward that essentially the same steps are needed to prove the
quasi-locality of the resulting charges. The reason for this is that
based on the inversion we can write the asymptotic identity
\begin{equation}
  \label{xmdef2}
   X_m (u) \approx (-\ii)\bar t_m(-u)\partial_u t_m (u),
\end{equation}
and here the derivation $\partial_u$ acts only locally: the resulting
operator will be translationally invariant and formally extensive. It
remains to be shown that the HS norm of the traceless part scales
linearly with the volume; this does not follow from \eqref{xmdef2},
and only holds for $u\in\physstr$.
For the computation of the HS norm one needs the same lattices
of size $2\times L$ and $4\times L$ which were constructed above. 
The proof of quasi-locality follows relatively easily once the leading
eigenvectors of the QTM's are found to be the delta states. This
procedure is described in
\cite{prosen-xxx-quasi,prosen-enej-quasi-local-review}, and we also
explain it in Appendix \ref{sec:proofs} with the technical details in the case of the
$SU(3)$-symmetric model.

We note that using the crossing relation \eqref{tmcrossing}  the
asymptotic inversion \eqref{su2inv} can be written in a form
equivalent to the l.h.s. of the fusion relation \eqref{su2t} with a
shifted rapidity. In this form the asymptotic inversion states that on
the r.h.s. the scalar (state independent) part will be dominant in the thermodynamic limit.

\begin{figure}
  \centering
  \begin{tikzpicture}
    \foreach \j in {0,1.2}
    {
      \begin{scope}[yshift=\j cm]
            \draw [thick,-<] (0,0) to (3.8,0);
    \draw [thick,<-] (3.8,0.6) to (0,0.6);
    \draw [thick,>->] (0,0) [out=180,in=270] to (-0.3,0.3) [out=90,in=180]
    to (0,0.6);
  \end{scope}
}
\node at (4.3,0) {$u$};
\node at (4.3,1.2) {$-u^*$};
    \foreach \i in {0,1,2,3}
    {
      \begin{scope}[xshift=\i cm]
      \draw [thick,->] (0.4,-0.4) to (0.4,2.2);      
      \end{scope}  
    }
    \node at (0.4,-0.8) {$1$};
    \node at (1.4,-0.8) {$2$};
        \node at (2.4,-0.8) {$\dots$};
    \node at (3.4,-0.8) {$L$};
    
    \begin{scope}[xshift=6cm]
      \node at (4.3,0) {$u$};
\node at (4.3,1.2) {$-u^*$};
   \foreach \j in {0,1.2}
    {
      \begin{scope}[yshift=\j cm]     
         \draw [thick,-<] (1,0) to (3.8,0);
    \draw [thick,<-] (3.8,0.6) to (1,0.6);
    \draw [thick,>->] (1,0) [out=180,in=270] to (0.7,0.3) [out=90,in=180]
    to (1,0.6);
      \end{scope}
}
    \foreach \i in {0,1,2,3}
    {
      \begin{scope}[xshift=\i cm]
      \draw [thick,->] (0.4,-0.4) to (0.4,2.2);      
      \end{scope}  
    }
     \node at (0.4,-0.8) {$1$};
    \node at (1.4,-0.8) {$2$};
        \node at (2.4,-0.8) {$\dots$};
    \node at (3.4,-0.8) {$L$};
    \end{scope}
  \end{tikzpicture}
  \caption{Graphical interpretation of the proof of the second equation in
    \eqref{su224}. The trace of the four transfer matrices in question
    is evaluated in the crossed channel by building a four-site Quantum
    Transfer Matrix, which adds one column to the diagram. A specific
    eigenvector of this QTM is the product of two delta-states: this
    state has eigenvalue 2, which follows from the local inversion as
    depicted in the graph, after taking the trace in the vertical direction. The desired identity in
    \eqref{su224} holds if this is the leading eigenvalue.  }
  \label{fig:localinversion2}
\end{figure}
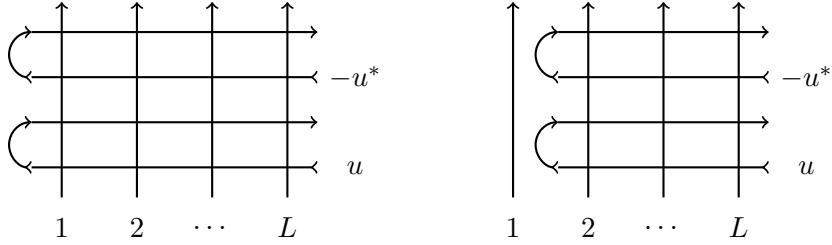

The inversion relation has important consequences for the transfer
matrix eigenvalues. Eq. \eqref{su2inv}
has to hold on the level of eigenvalues for {\it almost all}
states. We now analyze the implications of this, by computing the
products $\bar t_m(-u)t_m(u)$.

The eigenvalues of the space reflected TM could be
computed from the crossing relation \eqref{tmcrossing}, but there is a
more direct way. It is known that in Bethe Ansatz space reflection can be represented
by the change $\{\lambda\}_N\to\{-\lambda\}_N$. The TM eigenvalues
involve ratios of $Q$-functions with certain shifts.
Changing the sign of both $u$ and all the rapidities is equivalent to
changing the signs of all the shift parameters.  Thus the eigenvalues
of $\bar t_m(-u)$ are immediately found from \eqref{tmueig}:
\begin{equation}
  \label{tmueig2}
  \bar  t_m(-u)=
  \frac{Q_1\left(u+\ii\frac m2\right)Q_1\left(u-\ii\frac {m+2}2\right)}{Q_0\left(u-\ii\frac {m+1}2\right)}
\sum_{k=0}^m \frac{Q_0\left(u-\ii\frac {m+1}2+\ii k \right)}{Q_1\left(u-\ii\frac {m}2 +\ii k\right)Q_1\left(u-\ii\frac {m+2}2 +\ii k\right)}.
\end{equation}
Let us now consider the product of the eigenvalues $\bar
t_m(-u)t_m(u)$. Expanding the product we obtain a sum of $(m+1)^2$
terms, each of which involves ratios of $Q$-functions. Inspection
shows that among these $(m+1)^2$ terms there will be a single one
which gives identically 1; this will come from the summands with index
$k=0$. All the remaining terms are ratios of $Q$-functions, which are
 exponentially increasing or decreasing with $L$, depending on
 $u$. It follows from the inversion identity that if $u$ is within the
 physical strip, then all of these terms have to be exponentially
 decreasing. This also implies, that from the $(m+1)$ terms in the
 eigenvalues of $t_m(u)$
 the $k=0$ term has to be the leading one for almost all states if $u\in\physstr$.
We can thus write the explicit formula
 \begin{equation}
   t_m(u)\approx\frac{Q_1\left(u+\ii\frac m2\right)}{Q_1\left(u-\ii\frac
       m2\right)},\qquad
   u\in\physstr,\quad m=1,2,\dots
 \end{equation}
These relations play an essential role in establishing the
string-charge identities \eqref{su2sc0}, see the original papers \cite{JS-CGGE,jacopo-massless-1}.

We will show that similar steps are needed also in the
$SU(3)$-symmetric model. That model has a more complicated Bethe Ansatz solution
and corresponding fusion hierarchy of transfer matrices, nevertheless
the inversion relations take an identical form, and are equally important for the derivations of the
string-charge relations.

\section{The $SU(3)$-symmetric Lai-Sutherland model}

\label{sec:SU3}

In the $SU(3)$ symmetric Lai-Sutherland model the eigenstates are
constructed as excitations over the reference state, which can be
chosen for example as
$\ket{\emptyset}=\ket{1}\otimes\dots\otimes\ket{1}$, where $\{\ket{1},\ket{2},\ket{3}\}$ is an orthonormal basis of $\mathbb{C}^3$. One-particle excitations
are spin waves which carry an internal degree of freedom
corresponding to the polarization, thus carrying the defining
representation of $SU(2)$. General multi-particle excited states can be
considered as interacting spin waves, and each state is described by
the momenta of the particles and also by an auxiliary wave function
describing the orientation in the resulting internal space.

Corresponding to this physical picture, the Bethe states of this model
can be characterized by two sets of rapidity parameters
$\{\lambda_j\}_{j=1}^N$ and $\{\mu_j\}_{j=1}^M$. Here $N$ is the
number of physical particles and the $\lambda_j$ describe their
quasi-momenta. Further, the secondary (or magnonic) rapidities
$\{\mu_j\}_{j=1}^M$ describe the orientation in the internal space;  
 they can be understood as the Bethe rapidities of an auxiliary spin
 chain problem.
The $GL(3)$ global quantum numbers are $(L-N,N-M,M)$, and it is
required that $N\le 2L/3$ and $M\le N/2$.

The un-normalized real space wave functions can be written as
\begin{equation}
  \begin{split}
| \boldsymbol \lambda_N, \boldsymbol \mu_M\rangle = \sum_{1\leq n_1<\ldots <n_N\leq L}\ \sum_{1\leq m_1<\ldots <m_M\leq N}\ 
\sum_{\mathcal P \in \mathcal S^N} \left(\prod_{1\leq r < l \leq N}\frac{\lambda_{\mathcal P (l)}-\lambda_{\mathcal P (r)}-\ii}{\lambda_{\mathcal P (l)}-\lambda_{\mathcal P (r)}}\right)\nonumber\\
\times \braket{\boldsymbol m|\boldsymbol \lambda_{\mathcal P},  \boldsymbol \mu}
\prod_{r=1}^N\left(\frac{\lambda_{\mathcal{P}(r)}+\frac\ii2}{\lambda_{\mathcal{P}(r)}-\frac\ii2}\right)^{n_r}
\prod_{r=1}^{M}(E_{32})_{m_r}\prod_{s=1}^{N}(E_{21})_{n_{s}}\ket{\emptyset}\,,
\label{eq:bethe_state}    
  \end{split}
\end{equation}
where we used the elementary matrices $E_{ji}=\ket{j}\bra{i}$. The wave
function amplitudes are given by
\begin{equation}
  \begin{split}
\braket{\boldsymbol m|\boldsymbol \lambda_{\mathcal P}, \boldsymbol  \mu} &=
\sum_{\mathcal R \in \mathcal S^M} A(\boldsymbol \lambda_{\mathcal R})
\prod_{\ell=1}^M F_{\boldsymbol \lambda_{\mathcal P}}(\mu_{\mathcal R(\ell)}; m_\ell)\,,\\
F_{\boldsymbol \lambda}(\mu,s)&=\frac{-\ii }{\mu - \lambda_s - \frac\ii2}
\prod_{n=1}^{s-1} \frac{\mu -  \lambda_n + \frac\ii2}{\mu - \lambda_n - \frac\ii2}\,,\\
A(\lambda) &=\prod_{1\leq r < l \leq M} \frac{\mu_l-\mu_r- \ii}{\mu_l-\mu_r}\,.    
  \end{split}
\end{equation}
It follows from the periodicity of the wave function that the
two sets of rapidities are solutions to the following Bethe equations:
  \begin{equation}
    \begin{split}
 \left(\frac{\lambda_j+\frac{\ii}{2}}{\lambda_j-\frac{\ii}{2}}\right)^L
 &=
 \prod_{k=1,k\ne j}^N \frac{\lambda_j - \lambda_k +\ii}{\lambda_j -  \lambda_k -\ii}
 \prod_{k=1}^M \frac{\lambda_j - \mu_k -\frac{\ii}2}{\lambda_j - \mu_k +\frac{\ii}2} , \qquad j=1,\dots,N \label{eq:BE} \\
 1 &= \prod_{k=1}^M \frac{ \mu_j -\lambda_k-\frac{\ii}2}{\mu_j-\lambda_k +\frac{\ii}2} \prod_{k=1,k\ne j}^N \frac{\mu_j - \mu_k +\ii}{\mu_j - \mu_k -\ii}, \qquad j=1,\dots, M.
    \end{split}
  \end{equation}
The lattice momentum and the energy of a
Bethe state $\ket{\{\lambda_j\}_{j=1}^N,\{\mu_j\}_{j=1}^M}$ is given
by the sum of one particle momentum and energy, respectively: 
\begin{align}
  \label{PE}
 P &=\sum_{j=1}^N p(\lambda_j), & p(\lambda)& = \ii \log \left(\frac{\lambda_j+\frac{\ii}{2}}{\lambda_j-\frac{\ii}{2}}\right), \\
 E &=\sum_{j=1}^N \varepsilon(\lambda_j), & \varepsilon (\lambda)& =
            -   \frac{1}{\lambda^2+\frac{1}{4}}.
    \label{Edef}
\end{align}
Note that both the momentum and the energy depend only on the Bethe
roots of the first type. 

The eigenvalues of the fundamental transfer matrices \eqref{tdef} on Bethe states can be
computed as \cite{kulish-resh-glN}
\begin{align}
  \label{tu1}
 t (u)&= \frac{Q_1\left(u-\frac\ii2\right)}{Q_1\left(u+\frac\ii2\right)}+ \frac{Q_0(u) Q_1\left(u+\frac{3\ii}2\right) Q_2(u)}{Q_0(u+\ii) Q_1\left(u+\frac\ii2\right) Q_2(u+\ii)} + \frac{Q_0(u) Q_2\left(u+2\ii\right)}{Q_0(u+\ii) Q_2\left(u+\ii\right)},
\end{align}
where
\begin{equation}
 Q_0 (u) = u^L   
\end{equation}
and
\begin{equation}
 Q_1 (u) = \prod_{j=1}^N (u-\lambda_j) ,\qquad
 Q_2 (u) = \prod_{j=1}^M (u-\mu_j)  
\end{equation}
are the  $Q$-functions associated with the first and second level
Bethe roots.

The eigenvalues of the local charges $\mathcal{Q}_k$ can be computed from
\eqref{tu1} using the definition \eqref{Qdef}. It follows from the
presence of the $Q_0(u)$ factor in the second and third terms of
\eqref{tu1} that only the first term will contribute to the eigenvalue
of $\mathcal{Q}_k$ as long as $k<L$. This implies that the local
charges only depend on the first level Bethe roots
$\{\lambda_j\}_{j=1}^N$. This already indicates that the local charges
cannot give a full description of the states. Note that \eqref{Edef} can be derived 
immediately from \eqref{tu1}.

\subsection{String hypothesis and Thermodynamic Bethe Ansatz}

In the Lai-Sutherland model both the first and second type of rapidities can
form strings, and they form the same patterns in the complex plain.
A Bethe state with $M_n^{(1)}$ and $M_n^{(2)}$ number of $n$-strings
for the first and second type of particles thus consists of the rapidities
\begin{equation}
\begin{split}
\label{eq:SU3string}
 \lambda_\alpha^{n,\ell} &=\lambda_\alpha^n +\ii\left( \frac{n+1}2 -   \ell \right) +\delta_{1,\alpha}^{n,\ell}
                           , \quad \ell=1,\dots,n,\ \quad \alpha = 1, \dots, M_n^{(1)} \\
 \mu_\alpha^{n,\ell} &=\mu_\alpha^n +\ii\left( \frac{n+1}2 - \ell \right) +\delta_{2,\alpha}^{n,\ell}, \quad \ell=1,\dots,n,\ \quad \alpha = 1, \dots, M_n^{(2)}.
\end{split}
\end{equation}
Here, $\lambda_\alpha^n,\ \mu_\alpha^n\in\mathbb{R}$ are the string centers, and $\delta_{1,\alpha}^{n,\ell},\ \delta_{2,\alpha}^{n,\ell}$ are the string deviations, exponentially small in large volumes.
The total number of Bethe roots is computed simply as
\begin{equation}
 \sum_{n=1}^\infty n M^{(1)}_n = N, \qquad \sum_{n=1}^\infty n M^{(2)}_n = M.
\end{equation}

We will be interested in the thermodynamic limit, when
$N\rightarrow\infty$, $M\rightarrow\infty$, $L\rightarrow\infty$
while we keep the ratios $N/L$ and $M/L$ fixed. The string hypothesis
states that only Bethe roots in the form of \eqref{eq:SU3string}
contribute to thermodynamic behaviour.  
For the string centers we introduce the densities
$\rho_n^{(1)}(u),\ \rho_n^{(2)}(u)$, which are normalized such that in a large
volume $L$ the total number of $n$-strings of the first/second type
between rapidities $u$ and $u+\Delta u$ is
$L \rho_n^{(1/2)}(u)\Delta u$. Similarly we introduce the
hole densities $\rho_{h,n}^{(1)},\rho_{h,m}^{(2)}$.

In the thermodynamic limit the Bethe
equations  \eqref{eq:BE} can be transformed into a set of coupled linear integral equations: 
\begin{equation}
\begin{split}
  \label{root0}
 \rho_{t,n}^{(1)}(\lambda) &= a_n(\lambda)-\sum_{m=1}^\infty a_{n,m}\star\rho^{(1)}_m (\lambda) +\sum_{m=1}^\infty b_{n,m}\star\rho_m^{(2)}(\lambda)\\
 \rho_{t,n}^{(2)}(\lambda) &= -\sum_{m=1}^\infty a_{n,m}\star\rho^{(2)}_m (\lambda) +\sum_{m=1}^\infty b_{n,m}\star\rho_m^{(1)}(\lambda),
\end{split}
\end{equation}
where the total root densities are
\begin{equation}
 \rho_{t,n}^{(j)}(u) = \rho_{n}^{(j)}(u)+\rho_{h,n}^{(j)}(u),\qquad j=1,2.
\end{equation}
The kernels are
\begin{equation}
\begin{split}
 a_{n,m} (\lambda) &= (1-\delta_{nm})a_{|n-m|}(\lambda) +2a_{|n-m|+2}(\lambda)+\dots 2a_{n+m-2}(\lambda)+a_{n+m}(\lambda),\\
 b_{n,m} (\lambda) &= a_{|n-m|+1}(\lambda) +2a_{|n-m|+3}(\lambda)+\dots 2a_{n+m-1}(\lambda),
\end{split}
\end{equation}
and $a_n(\lambda)$ is defined in \eqref{andef}.

Similarly to the spin-$1/2$ case, the aforementioned equations
can be cast in a partially decoupled form (for derivation, see
e.g. \cite{nested-quench-1}): 
\begin{equation}
\begin{split}
  \label{eq:root1}
 \rho_{t,n}^{(1)}(\lambda) &= \delta_{n,1}s(\lambda) + s\star\left(\rho_{h,n-1}^{(1)} + \rho_{h,n+1}^{(1)} \right)(\lambda) + s\star\rho_n^{(2)}(\lambda) \\
 \rho_{t,n}^{(2)}(\lambda) &= s\star\left(\rho_{h,n-1}^{(2)} + \rho_{h,n+1}^{(2)} \right)(\lambda) + s\star\rho_n^{(1)}(\lambda),
\end{split}
\end{equation}
where the following definitions and conventions are understood:
\begin{align}
 \rho_{h,0}^{(r)}(u) &= 0, \quad r=1,2,
\end{align}
and $s(u)$ is defined in \eqref{sdef}.
The structure  of these integral equations is depicted on the ``TBA diagram''
 \ref{fig:su3tba}.

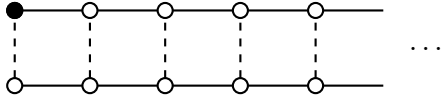
\begin{figure}
  \centering
  \begin{tikzpicture}
    \draw[thick,fill] (0,1) circle (0.1);
    \foreach \i in {0,1,2,3,4}
          {
      \begin{scope}[xshift=\i cm]
      \draw[thick] (0,0) circle (0.1);       
      \draw[thick] (0.1,0) to (0.9,0);
        \draw[thick] (0,1) circle (0.1);       
        \draw[thick] (0.1,1) to (0.9,1);
        \draw[thick,dashed] (0,0.1) to (0,0.9);
    \end{scope}
  }
  \node at (5.5,0.5) {$\dots$};
  \end{tikzpicture}
  \caption{TBA diagram of the $SU(3)$-symmetric model. The node in row
    $a$ and column $m$ from the top left correspond to the $m$-string
    of particle type $a=1,2$.
The links denote the two different convolutions in
    Eq. \eqref{eq:root1}. The filling of the top left node signals the
    source term for $\rho^{(1)}_{t,1}(u)$.}
  \label{fig:su3tba}
\end{figure}

\subsection{Strategy towards the string-charge relations}

It is our goal to find a set of quasi-local charges, which will
uniquely determine all root densities $\rho^{(1)}_m(u)$ and
$\rho^{(2)}_m(u)$. Based on the XXX chain it is a natural
idea to consider the fusion hierarchy of the transfer matrices.

In the $SU(3)$-symmetric model there is a set of fused transfer matrices
$t_m^{(a)}(u)$ with $a=1,2$, and $m=1,2,\dots,\infty$. They correspond
to representations of $SU(3)$ described by the rectangular Young diagram with $a$
rows and $m$ columns.
Precise definitions of $t_m^{(a)}(u)$ using local Lax operators will
be given in Section 
\ref{sec:SCrelations}.
We put forward that these TM's satisfy the Hirota equation or $T$-system
 \cite{zabrodin-hirota-1,suzuki-kuniba-tomoki-t-system-y-system-review}
\begin{equation}
\begin{split}
  \label{su3t}
  t^{(a)}_m\left(u+\frac\ii2\right)t^{(a)}_m\left(u-\frac\ii2\right)&=t^{(a)}_{m+1}(u)t^{(a)}_{m-1}(u)+
t^{(a-1)}_{m}(u)t^{(a+1)}_{m}(u), \\ \qquad a=1,2,&\quad m=1,2,\dots
\end{split}
\end{equation}
Here $t^{(a)}_0(u)=1$ by definition and $t^{(0)}_m(u)$ and $
t^{(3)}_m(u)$ are scalar functions that will be specified later.

The structure of this $T$-system is very closely related to the ``TBA diagram''
\ref{fig:su3tba}. The top and bottom rows correspond to the transfer
matrices with $a=1$ and $a=2$, respectively. 
Based on previous experience we
expect to find a relation between the eigenvalues of these fused transfer
matrices and the Bethe root densities. We put forward the formal definition
\begin{equation}
  X^{(a)}_m(u)=(-\ii)\frac{\partial}{\partial u}\log t^{(a)}_m(u).
\end{equation}
We will show that the $X^{(a)}_m(u)$  can be constructed locally, and
their mean values are connected to the
root densities by relations analogous to \eqref{su2sc}. Furthermore,
we will prove in two specific cases that they are quasi-local within the physical
strip.

\section{Quasi-local charges: the defining and conjugate
  representations}
\label{sec:DefAndConjRep}

In this section we investigate two special transfer matrices in the
$SU(3)$-case: the operators $t^{(1)}_1(u)$ and $t^{(2)}_1(u)$ that
correspond to the defining and conjugate representations of $SU(3)$, respectively.

The first one is the fundamental TM, corresponding
to the defining representation with highest weight
$(1,0)$ and the Young diagram with one box. For this TM we will
use the notation $t^3(u)$, and the definition is the same as in
\eqref{tdef}:
\begin{equation}
  \label{tdefsu3}
 t^3(u) = \Tr_0 R_{10}(u)R_{20}(u)\dots R_{L0}(u),
\end{equation}
where $R(u)$ is given by \eqref{Rdef} with $N=3$.

The second special TM corresponds to the conjugate (or anti-symmetric
tensor) representation of
$SU(3)$, given by highest weight $(0,1)$ and the Young diagram with
two rows and one column. We use the notation $t^{\bar 3}(u)$ and the
definition 
\begin{equation}
   t^{\bar 3}(u) = \Tr_0 R^{\bar 3}_{10}(u)R^{\bar 3}_{20}(u)\dots R^{\bar 3}_{L0}(u),
\end{equation}
where $R^{\bar 3}(u)$ is the $R$-matrix acting on the tensor product
of a fundamental and conjugate representation. It is given by
\begin{align}
\label{eq:R3barDef}
  R^{3,\bar{3}}(u) &\equiv R^{\bar{3}} (u)
                     = \frac{u+\frac{3\ii}{2}-\ii K} {u+\frac{3\ii}{2}}.
\end{align}
Here $K$ is the trace operator (or Temperley-Lieb operator) with
matrix elements $K_{ab}^{cd} =\delta_{ab}\delta_{cd}$. Note that $K$
is the partial transpose of the permutation operator: $P^{t_1} =
P^{t_2} = K$. Therefore we have the simple relation
\begin{equation}
\label{eq:SpaceReflectionR}
 R^{\bar 3} (u)= \frac{u+\frac\ii2}{u+\frac{3\ii}2} R^{t_1}\left(-u-\tfrac{3\ii}2\right).
\end{equation}
This $R$-matrix also satisfies the unitarity relation:
\begin{equation}
\label{eq:bar3uni}
 R^{\bar 3}(u) R^{\bar 3}(-u) = 1,
\end{equation}
which is easily checked using the Temperley-Lieb property $K^2 = 3 K$. 

For this special conjugate $R$-matrix the compatibility conditions are
easily derived from the original Yang-Baxter relation
\eqref{YB}. Making use of \eqref{eq:SpaceReflectionR} and
\eqref{eq:bar3uni} it is easy to show that
\begin{equation}
  \begin{split}
  \label{YB2}
 R_{23}(v-w)\bar R_{13}(u-w)\bar R_{12}(u-v) &=
 \bar R_{12}(u-v)\bar R_{13}(u-w)R_{23}(v-w).
  \end{split}
\end{equation}
It follows that the TM's $\{t^3_m(u),t^{\bar 3}_m(u)\}$ are all
commuting with each other.

The group invariance property of the conjugate transfer matrix follows
easily from \eqref{groupinvariance}. Taking a partial transpose in the
second space we get 
\begin{equation}
  \label{group2}
  G_1(G_2^{t})^{-1} R^{t_2}(u)=R^{t_2}(u) G_1(G^t_2)^{-1},\qquad G\in GL(3).
\end{equation}
For any $G\in SU(3)$ we have $(G^t)^{-1}=G^*$, thus we see by the
relation \eqref{eq:SpaceReflectionR}
that the matrix $R^{\bar 3}(u)$ is compatible with representations
$3\otimes \bar 3$.

Similar to the case of the XXX chain it is very useful to define the
space reflected transfer matrices corresponding to these two representations:
\begin{align}
  \bar{t}^3 (u) &= \Tr_0 R_{L0}^3(u)\dots R_{10}^3(u) \\
  \bar{t}^{\bar 3} (u) &= \Tr_0 R_{L0}^{\bar3}(u)\dots R_{10}^{\bar3}(u).
\end{align}
The local crossing relation \ref{eq:SpaceReflectionR} implies the
following connections:
\begin{equation}
  \label{su3crossing}
  \begin{split}
    \bar{t}^{\bar3} (u) &= \left(\frac{u+\tfrac{\ii}2}{u+\tfrac{3\ii}2}\right)^L t^{3}\left(-u-\tfrac{3\ii}2\right), \\
 \bar{t}^3(u) &= \left(\frac{u}{u+\ii}\right)^L  t^{\bar3}\left(-u-\tfrac{3\ii}2\right).
  \end{split}
\end{equation}
These are simple generalizations of the crossing relation \eqref{tmcrossing} in
the $SU(2)$ case. Even though the space reflected operators are not
independent,  we keep their definition and this special notation,
because they are very useful to study the
inversion relations and quasi-locality properties of the charges.

For example, the adjoints of the transfer matrices can be expressed
simply as
\begin{equation}
  \label{su3adjoint}
  (t^3)^\dagger(u)=   \bar t^3(-u^*)\qquad
  (t^{\bar 3}(u))^\dagger=\bar t^{\bar 3}(-u^*).
\end{equation}
For the details see Lemma \ref{adjointlem} in the Appendix.

\begin{thm}
For $t^3 (u)$ and $t^{\bar 3} (u)$ the following asymptotic inversion
relations hold when $u\in\physstr$ 
\begin{equation}
    \label{eq:inv1}
    \begin{split}
      \bar t^3(-u)  t^3(u)&\approx 1
      \\
  \bar t^{\bar 3}(-u) t^{\bar 3}(u)&\approx 1.
    \end{split}
\end{equation}
\end{thm}

A rigorous proof is presented in Appendix \ref{sec:proofs}. The
proof is based on the same ideas as explained in the case of the XXX
model. The core relations are the local inversions
\eqref{eq:uni1}-\eqref{eq:bar3uni}, which guarantee the exact
inversion of the {\it monodromy matrices}. Considering the transfer
matrices, the ``boundary effect'' of taking the trace does not
propagate into the bulk of the chain, if $u$ is chosen from the
physical strip.

Based on the results of the $SU(2)$-symmetric chain we define the
following two
generating functions for the charges:
\begin{equation}
  \begin{split}
  \label{XY1}
  X(u)&= (-\ii)\partial_u \log t(u)\\
  Y(u)&= (-\ii)\partial_u \log t^{\bar 3}(u).
  \end{split}
\end{equation}
Due to the asymptotic inversion they can be written in large enough
volumes as
\begin{equation}
  \label{XY2}
  \begin{split}
  X(u)&\approx(-\ii) \bar t(-u)\partial t(u)\\
  Y(u)&\approx(-\ii) \bar t^{\bar 3}(-u)  \partial t^{\bar 3}(u).
  \end{split}
\end{equation}
The formula \eqref{XY1} is convenient for the treatment of the
eigenvalues, whereas the advantage of \eqref{XY2} is its local
construction, enabling the evaluation of mean values in initial
states and the proof of quasi-locality.

It follows from the asymptotic inversion and the adjoint property
\eqref{su3adjoint} that for $u\in\valos$ these operators are Hermitian.

\begin{thm}
  The traceless operators $\{X(u)\}$ and $\{Y(u)\}$ are quasi-local within
  the physical strip.
\end{thm}
The detailed proof of this Theorem is presented in appendix
\ref{sec:proofs}. The proof uses the same techniques as in the
original works \cite{prosen-xxx-quasi,prosen-enej-quasi-local-review},
and the starting point is the existence of local inversion relations
for the $R$-matrices.

Let us now investigate the consequences of the inversion relation for
the eigenvalues of these transfer matrices.

For the fundamental TM and its space reflected counterpart we have 
\begin{equation}
\begin{split}
 t^3 (u)&= \frac{Q_1\left(u-\frac\ii2\right)}{Q_1\left(u+\frac\ii2\right)}+ \frac{Q_0(u) Q_1\left(u+\frac{3\ii}2\right) Q_2(u)}{Q_0(u+\ii) Q_1\left(u+\frac\ii2\right) Q_2(u+\ii)} + \frac{Q_0(u) Q_2\left(u+2\ii\right)}{Q_0(u+\ii) Q_2\left(u+\ii\right)}   \\
 \bar t^3 (-u)&= \frac{Q_1\left(u+\frac\ii2\right)}{Q_1\left(u-\frac\ii2\right)}+ \frac{Q_0(u) Q_1\left(u-\frac{3\ii}2\right) Q_2(u)}{Q_0(u-\ii) Q_1\left(u-\frac\ii2\right) Q_2(u-\ii)} + \frac{Q_0(u) Q_2\left(u-2\ii\right)}{Q_0(u-\ii) Q_2\left(u-\ii\right)},   \\
\end{split}
\end{equation}
where the second equality follows simply from the fact that in 
Bethe Ansatz the space reflection is described by negating all Bethe rapidities.

Using the relations \eqref{su3crossing} we obtain further
\begin{equation}
\begin{split}
 t^{\bar 3}(u) &=
  \frac{Q_0\left(u+\frac\ii2\right) Q_1\left(u+2\ii\right) }{Q_0\left(u+\frac{3\ii}2\right) Q_1\left(u+\ii\right) } + \frac{Q_1(u) Q_2\left(u+\frac{3\ii}2\right)}{Q_1(u+\ii) Q_2\left(u+\frac\ii2\right) } + \frac{Q_2\left(u-\frac\ii2\right) }{Q_2\left(u+\frac\ii2\right) }  \\
 \bar t^{\bar 3}(-u) &=
  \frac{Q_0\left(u-\frac\ii2\right) Q_1\left(u-2\ii\right) }{Q_0\left(u-\frac{3\ii}2\right) Q_1\left(u-\ii\right) } + \frac{Q_1(u) Q_2\left(u-\frac{3\ii}2\right)}{Q_1(u-\ii) Q_2\left(u-\frac\ii2\right) } + \frac{Q_2\left(u+\frac\ii2\right) }{Q_2\left(u-\frac\ii2\right) }.  \\
\end{split}
\end{equation}

Let us now investigate the inversion relations \eqref{eq:inv1} on the
level of these eigenvalues. Multiplying the sums we observe that in
both cases there will be 9 terms out which only one is identically
equal to 1. The remaining 8 terms will include various ratios of
$Q$-functions. In the large volume limit these ratios will be
exponentially increasing or decreasing, depending on $u$. It follows
from the inversion relation, that within the individual sums those
terms have to be dominant for $u\in\physstr$ which produce the
required identity. We thus have the relations for $u\in\physstr$
\begin{equation}
  \label{su3tdl1}
  \begin{split}
     t^3 (u)&\approx
     \frac{Q_1\left(u-\frac\ii2\right)}{Q_1\left(u+\frac\ii2\right)}\\
       t^{\bar 3}(u) &\approx
       \frac{Q_2\left(u-\frac\ii2\right) }{Q_2\left(u+\frac\ii2\right)  }.\\
  \end{split}
\end{equation}
These are crucial in establishing the thermodynamic limit of the
charges and eventually the string-charge relations.

Notice the symmetry of these relations: exchanging the defining and conjugate
representations simply corresponds to exchanging the two types of
Bethe rapidities. This is simply the conjugation symmetry of the
Dynkin diagram of $SU(3)$, which is nicely reflected by the Bethe
Ansatz solution (compare with Fig \ref{fig:su3tba}).

It is important that \eqref{su3tdl1} does not necessarily hold for
{\it all} Bethe states. For example, if we choose $u\in\valos$ and keep the number of
rapidities finite while performing the $L\to\infty$ limit then there
will be two remaining finite terms in $t^{\bar 3}(u)$. Our proof using
the inversion relation only tells us that \eqref{su3tdl1} will hold
for almost all states, where the probability measure is derived simply
from the HS scalar product. This corresponds to the infinite
temperature thermal ensemble. Therefore, in the first instance our
statement only concerns the infinite temperature Bethe states.
Nevertheless, it can be argued based on continuity that \eqref{su3tdl1} still holds for Bethe root densities
``close'' to the infinite temperature state. We give further comments
on this issue at the end of the next Section.

\section{Arbitrary representations and string-charge relations}

\label{sec:SCrelations}

In this section we treat all the representations of $SU(3)$ that
correspond to rectangular Young diagrams. To this order we define the
families $t^3_m(u)\equiv t^{(1)}_m(u)$ and $t^{\bar 3}_m(u)\equiv
t^{(2)}_m(u)$ with $m=1,2,\dots$. For $m=1$ they coincide with the two 
transfer matrices $t^3(u)$ and $t^{\bar 3}(u)$ of the previous Section.

For a representation of $\Lambda$ of $GL(N)$ with rectangular Young diagrams the $R$-matrix
acting on the tensor product of the defining representation and
$\Lambda$ can be expressed as \cite{kulish-resh-sklyanin--fusion,drinfeld-hopf,yangian-hep-intro}
\begin{equation}
  \label{egylg}
  R^\Lambda(u)\quad\sim\quad u+\ii \sum_{i,j}E_{ij}\Lambda_{ji},
\end{equation}
where $E_{ij}$ are the elementary matrices acting on the defining
representation, $\Lambda_{ij}$ are their
representations, and there is an arbitrary normalization factor not
included in \eqref{egylg}. If $\Lambda$ is the defining
representation, we get back the usual $R$-matrix, as
$P=\sum_{i,j}E_{ij}E_{ji}$. 

Let $\Lambda^{(m,0)}_{ij}$ and $\Lambda^{(0,m)}_{ij}$ be the
representation matrices of $GL(3)$ corresponding to the $1\times m$
Young diagram with $GL(3)$ highest weight $(m,0,0)$ and to the $2\times m$ Young
diagram with $GL(3)$ highest weight $(m,m,0)$, respectively. For these two
families of representations we introduce the normalized $R$-matrices
\begin{equation}
  \begin{split}
R^{(m,0)}(u)&=  \frac{ u-\ii\frac{m-1}{2}+\ii \sum_{i,j} E_{ij}\Lambda^{(m,0)}_{ji}}{u+\ii\frac{m+1}{2}}  \\
  R^{(0,m)}(u)&=  \frac{ u-\ii\frac{m-2}{2}+\ii \sum_{i,j} E_{ij}\Lambda^{(0,m)}_{ji}}{u+\ii\frac{m+2}{2}}.  \\  
  \end{split}
  \label{su3R}
\end{equation}
Compared to \eqref{egylg} they involve a scalar factor and a shift in
the rapidity, which is equivalent to adding a constant to all
components of the weight vector.

These $R$-matrices satisfy the local inversion relations
\begin{equation}
  \label{su3teljeslocalinv}
  R^{(m,0)}(u)R^{(m,0)}(-u)=1, \qquad
   R^{(0,m)}(u) R^{(0,m)}(-u)=1.
 \end{equation}
This is a general property of $R$-matrices
\cite{drinfeld-hopf,yangian-hep-intro}, but we also prove it
explicitly in Appendix \ref{sec:Rinv}.

 Now we define the transfer matrices
 \begin{equation}
   \label{su3rect}
   \begin{split}
   t^3_m(u)&= \text{Tr}_a   R^{(m,0)}_{1,a}(u)  R^{(m,0)}_{2,a}(u)\dots R^{(m,0)}_{L,a}(u)\\
t^{\bar 3}_m(u)&= \text{Tr}_a   R^{(0,m)}_{1,a}(u) R^{(0,m)}_{2,a}(u)\dots R^{(0,m)}_{L,a}(u). \\
 \end{split}
   \end{equation}

In analogy with the previous Sections we also define the space
reflected transfer matrices
 \begin{equation}
   \begin{split}
   \bar t^3_m(u)&= \text{Tr}_a   R^{(m,0)}_{L,a}(u)  R^{(m,0)}_{L-1,a}(u)\dots R^{(m,0)}_{1,a}(u)\\
  \bar t^{\bar 3}_m(u)&= \text{Tr}_a   R^{(0,m)}_{L,a}(u)  R^{(0,m)}_{L-1,a}(u)\dots R^{(0,m)}_{1,a}(u).\\
 \end{split}
   \end{equation}
They are not independent from the TM's in \eqref{su3rect}, they are
related by crossing transformations of the type
\eqref{su3crossing}. However, these relations are not relevant for our purposes.

Based on our earlier results we formulate the following:

\begin{conj}
   For $u\in\physstr$ the
following asymptotic inversion relations hold:
\begin{equation}
  \label{tm3i}
\begin{split}
  \bar{t}_m^3 (-u) t_m^3 (u) & \approx 1 \\
  \bar{t}_m^{\bar 3} (-u) t_m^{\bar 3} (u) & \approx 1.
\end{split} 
\end{equation}
\end{conj}
The main idea behind this is the same as earlier: the global inversion of {\it
  monodromy} matrices follows from the local inversion relations
\eqref{su3teljeslocalinv}. Then the asymptotic inversion of transfer
matrices will hold for $u\in\physstr$ if the ``boundary effect'' of
taking the trace does not propagate into the bulk of the chain. This
could be checked by constructing the appropriate 2-site and 4-site
Quantum Transfer Matrices, in analogy with the computations of the
previous Sections. At present we do not have a general
proof, except for the two special cases with $m=1$, presented above
and in Appendix \ref{sec:proofs}.

We define the following two families of generating functions for
quasi-local charges:
\begin{equation}
\begin{split}
\label{eq:XmYmDef}
 X_m (u) &= (-\ii)\partial_u \log t^3_m (u)\\
 Y_m (u) &= (-\ii)\partial_u \log t^{\bar 3}_m (u).
\end{split}
\end{equation}
Under the assumption of the above conjecture, they are asymptotically
equal to the following locally constructed quantities: 
\begin{equation}
\begin{split}
\label{eq:XmYmDef2}
 X_m (u) &\approx (-\ii) t^3_m (-u) (\partial t^3_m) (u)\\
 Y_m (u) &\approx (-\ii) t^{\bar3}_m (-u) (\partial t^{\bar 3}_m) (u).
\end{split}
\end{equation}

\begin{conj}
For $u\in\physstr$ the traceless operators $\{X_m(u)\}$ and
$\{Y_m(u)\}$ are quasi-local.
\end{conj}
We do not have a proof for this conjecture, but the cases of
$\{X_1(u)\}$ and $\{Y_1(u)\}$ we treated in the previous section,
 together with the known results for the $SU(2)$ case 
already give strong motivation for its validity.

Let us now treat the eigenvalues of the transfer matrices and the charges.
The eigenvalues of $t^3_m(u)$ and $t^{\bar 3}_m(u)$ can be expressed
using the $Q$-functions as
\begin{equation}
\begin{split}
\label{eq:ttbarQ}
 t^{3}_m (u) &= \frac{Q_1\left(u-\ii\frac m2\right) Q_2\left(u+\ii\frac{m+3}2\right)}{Q_0\left(u+\ii\frac{m+1}2\right)} \sum_{k=0}^m \frac{Q_0\left(u+\ii\frac{m+1}2-\ii k\right) Q_2\left(u+\ii\frac{m+1}2-\ii k\right)}{Q_1\left(u+\ii\frac{m}2-\ii k\right) Q_1\left(u+\ii\frac{m+2}2-\ii k\right)} \times\\
  & \qquad\qquad\qquad\qquad\qquad\qquad \times\sum_{\ell=0}^k \frac{Q_1\left(u+\ii\frac{m+2}{2}-\ii\ell\right)}{Q_2\left(u+\ii\frac{m+1}{2}-\ii\ell\right) Q_2\left(u+\ii\frac{m+3}{2}-\ii\ell\right)} \\
 t^{\bar 3}_m (u) &= \frac{Q_1\left(u+\ii\frac{m+3}2\right) Q_2\left(u-\ii\frac{m}2\right)}{Q_0\left(u+\ii\frac{m+2}2\right)} \sum_{k=0}^m \frac{Q_0\left(u-\ii\frac{m-2}2+\ii k\right) Q_2\left(u-\ii\frac{m-2}2+\ii k\right)}{Q_1\left(u-\ii\frac{m-3}2+\ii k\right) Q_1\left(u-\ii\frac{m-1}2+\ii k\right)}\times \\
 & \qquad\qquad\qquad\qquad\qquad\qquad \times\sum_{\ell=0}^k \frac{Q_1\left(u-\ii\frac{m-1}2+\ii\ell\right)}{Q_2\left(u-\ii\frac{m}2+\ii\ell\right) Q_2\left(u-\ii\frac{m-2}2+\ii\ell\right)}.
\end{split}
\end{equation}
These explicit formulas can be derived from the more general ``tableaux
sum'' valid in the $SU(N)$-symmetric model
\cite{bazhanov-reshetikhin-rsos-fusion}.
The concrete formula for the general ``tableaux sum'' will be given in the next Section.

It can be checked by direct substitution that
these eigenvalue functions solve the Hirota equation \eqref{su3t} with
boundary conditions $t^{(a)}_0=1$ and
\begin{equation}
  \label{su3fura}
  t^{(0)}_m(u)=\frac{Q_0(u-\ii \frac{m}{2})}{Q_0(u+\ii \frac{m}{2})},
  \qquad t^{(3)}_m(u)=1.
\end{equation}
These boundary conditions might seem somewhat unnatural: they differ
from the most often used conventions, see for example the comparison
on page 18. of \cite{zabrodin-hirota-1}. We apply this normalization
so that the inversion relations hold without additional
``kinematical'' $Q_0$ factors. Also, this normalization is most
convenient for the string-charge relations.

The matching of the formulas \eqref{eq:ttbarQ} with the
normalization of the Lax matrices \eqref{su3R} is checked easily by
computing the eigenvalues on the reference state
$\ket{\emptyset}=\ket{111\dots 1}$ . This amounts to
setting $Q_1(u)=Q_2(u)=1$, and comparing the remaining ratios of $Q_0(u)$
functions to the direct application of \eqref{su3R} using the
eigenvalues of the representation matrices $\Lambda_{11}^{(m,0)}$ and $\Lambda_{11}^{(0,m)}$.

In order to derive the string-charge relations, we are interested in
the thermodynamic limit of the above eigenvalues. We
will argue that in the thermodynamic limit for $u\in\physstr$ the leading terms in the above sums are: 
\begin{equation}
\begin{split}
\label{eq:t3tbar3TDL}
 t^{3}_m (u) &\approx \lim_{\text{TDL}} \frac{Q_1 \left(u-\ii\frac{m}{2}\right)}{Q_1 \left(u+\ii\frac{m}{2}\right)} \\
  t^{\bar 3}_m (u) &\approx \lim_{\text{TDL}} \frac{Q_2 \left(u-\ii\frac m2\right)}{Q_2 \left(u+\ii\frac m2\right)},
\end{split}
\end{equation}
which respectively correspond to the $k=0,\ \ell=0$ and the $k=m,\ \ell=m$ terms. 

Our reasoning is the same as in the previous Sections: we are
selecting those terms from the eigenvalues which automatically produce
the asymptotic inversions \eqref{tm3i}. The eigenvalues of the space
reflected TM's evaluated at $-u$ are given by formally
the same expressions involving the ratios of $Q$-functions, with the
signs of the various shifts reversed. After some inspection it can be
seen that from the various terms in \eqref{eq:ttbarQ} only the ratios 
given in \eqref{eq:t3tbar3TDL} lead to the
inversion relations.
If they would not be the leading terms, then the
asymptotic inversion could not hold, and this proves their dominance.

Let us now evaluate the mean values of the charges $X_m(u)$,
$Y_m(u)$.  We start from a finite volume Bethe state with rapidities
$\{\lambda\}_{N_1},\{\mu\}_{N_2}$. 

Assuming that the dominant terms are given by \eqref{eq:t3tbar3TDL} we
have for the eigenvalues
\begin{equation}
\begin{split}
  X_m (u) &=
   -\ii \lim_{\text{TDL}} \partial_u \log \frac{Q_1  \left(u-\ii\frac{m}{2}\right)}{Q_1 \left(u+\ii\frac{m}{2}\right)}
  =
   -\ii\lim_{\text{TDL}} \sum_{j=1}^{N_1}\left( \frac{1}{u-\lambda_j-\ii\frac{m}{2}} - \frac{1}{u-\lambda_j+\ii\frac{m}{2}}\right)\\
  Y_m (u) &=
  -\ii\lim_{\text{TDL}}\partial_u \log \frac{Q_2  \left(u-\ii\frac{m}{2}\right)}{Q_2 \left(u+\ii\frac m2\right)}
  = -\ii\lim_{\text{TDL}}\sum_{j=1}^{N_2} \left(\frac{1}{u-\mu_j-\ii\frac m2} - \frac{1}{u-\mu_j + \ii\frac m2} \right).
\end{split}
\end{equation}
These formulas refer to the exact Bethe roots. Making use of the
string hypothesis in the TDL we get the expressions
\begin{equation}
\begin{split}
  \label{XYkoztes1}
 X_m(u) &= 2\pi L\sum_{n=1}^\infty \int_{-\infty}^\infty d\lambda\, \rho_n^{(1)}(\lambda) \sum_{j=1}^{\min(m,n)} a_{|n-m|-1+2j}(u-\lambda) = \\ &\qquad = 2\pi\sum_{n=1}^\infty \sum_{j=1}^{\text{min}(m,n)} (a_{|n-m|-1+2j} \star \rho_n^{(1)})(u) \\
  Y_m(u) &= 2\pi L\sum_{n=1}^\infty \int_{-\infty}^\infty d\lambda\, \rho_n^{(2)}(\lambda) \sum_{j=1}^{\min(m,n)} a_{|n-m|-1+2j}(u-\lambda) = \\ &\qquad = 2\pi\sum_{n=1}^\infty \sum_{j=1}^{\text{min}(m,n)} (a_{|n-m|-1+2j} \star \rho_n^{(2)})(u).
\end{split}
\end{equation}
Here the summation runs over the possible $n$ strings. We made use
of the identity
\begin{equation}
\begin{split}
\sum_{\ell=1}^n &\frac{1}{u-\left(\lambda+\ii\left(\frac{n+1}{2}-\ell\right)\right)-\ii\frac{m}{2}} - \frac{1}{u-\left(\lambda+\ii\left(\frac{n+1}{2}-\ell\right)\right)+\ii\frac{m}{2}} = \\ =& 2\pi\ii\sum_{j=1}^{\min(n,m)} a_{|n-m|-1+2j} (u-\lambda).
\end{split}
\end{equation}
The formulas \eqref{XYkoztes1} can be transformed into more compact
forms using standard tricks. This lengthy, but straightforward
computation is delegated to appendix \ref{sec:tbasc}, and the results are the following.

The hole densities can be expressed using the charges as 
\begin{equation}
  \label{eq:SCsu3toki}
\begin{split}
 \rho_{h,m}^{(1)}&= a_m - \frac{1}{2\pi L}\left( X_m^{[+]} + X_m^{[-]} - Y_m \right) \\
 \rho_{h,m}^{(2)}&= - \frac{1}{2\pi L}\left( Y_m^{[+]} + Y_m^{[-]} - X_m \right).
\end{split}
\end{equation}
Inverting this relation we find
\begin{equation}
  \begin{split}
  \label{eq:SCdualityV2itt}
  \frac{1}{2\pi L}   X_m&=-\left(G_1\star  (\rho_{h,m}^{(1)}+a_m)+G_2 \star \rho_{h,m}^{(2)}\right)\\
  \frac{1}{2\pi L}   Y_m&=-\left(G_1\star \rho_{h,m}^{(2)}+G_2 \star (\rho_{h,m}^{(1)}+a_m)\right)
  \end{split}
\end{equation}
with the kernels
\begin{equation}
  \begin{split}
    G_1(x)&=\int \frac{dk}{2\pi}e^{-\ii kx}\hat G_1(k)
=\frac{1}{\sqrt{3}}
\frac{\cosh(\pi/3 x)}{\cosh(\pi x)}=\frac{1}{\sqrt{3}}
   \frac{1}{2\cosh(2\pi x/3)-1} \\
    G_2(x)&=\int \frac{dk}{2\pi}e^{-\ii kx}\hat G_2(k)=\frac{1}{\sqrt{3}}
    \frac{\sinh(\pi/3 x)}{\sinh(\pi x)}=\frac{1}{\sqrt{3}}
    \frac{1}{2\cosh(2\pi x/3)+1}
\end{split}
\end{equation}
Finally, the root densities can be expressed as 
\begin{equation}
  \begin{split}
  \label{eq:SU3SCdualityV3}
     \rho_{m}^{(1)}&=\frac{1}{2\pi L}\left(X_m^{[+]} + X_m^{[-]}-X_{m-1}-X_{m+1} \right)
     \\
     \rho_{m}^{(2)}&=\frac{1}{2\pi L}\left(Y_m^{[+]}+Y_m^{[-]}-Y_{m-1}-Y_{m+1}\right).
 \end{split}
\end{equation}
This is an immediate generalization of the string-charge
relation \eqref{su2sc} of the $SU(2)$-chain. Once again we can observe
the symmetry of the Dynkin diagram of $SU(3)$: exchanging the defining
and conjugate representations is mirrored by the exchange of the two
Bethe rapidity types.

The crucial point in our derivation was selecting the dominant term in
the transfer matrix eigenvalues. The argument based on the asymptotic
inversion only holds for the infinite temperature states, and some
neighborhood of these root distributions. 
At present we can not exclude the existence of Bethe root densities,
which would select a different term, thus violating
\eqref{eq:SU3SCdualityV3}. Nevertheless we performed an independent check in a
particular case, namely for the quench problem with initial state
\begin{equation}
   \ket{\Psi_\delta}=\prod_{j=1}^{L/2} \frac{\ket{11}+\ket{22}+\ket{33}}{\sqrt{3}}
\end{equation}
This quench was studied in \cite{sajat-su3-1,sajat-su3-2} where the
exact root densities were determined using Boundary Quantum Transfer
Matrix methods. Now we computed the mean values of the
first two members 
$X_1(u)$ and $Y_1(u)$ in this initial state and we checked that the relations 
\eqref{eq:SU3SCdualityV3} indeed hold.  This computation is presented
in Appendix \ref{sec:deltastate}.

\section{Generalization to $SU(N)$}
\label{sec:SUNGGE}
In this section, we consider the generalization of the previous
results for $SU(N)$ spin chains. The construction laid out here is a
direct generalization of the case of $SU(3)$. However, most of our
statements here are conjectures, motivated by the earlier
results.

In the $SU(N)$ case the nested Bethe Ansatz
involves $(N-1)$ sets of Bethe rapidities, which will be denoted as
$\{\{\lambda^{(a)}_j\}_{j=1,\dots,N_a}\}_{a=1,\dots,N-1}$. Correspondingly,
the $Q$-functions of the model are
\begin{equation}
Q_0(u)=u^L,\qquad  Q_a(u)=\prod_{j=1}^{N_a} (u-\lambda^{(a)}_j),\qquad
Q_N(u)=1.
\end{equation}
The eigenvalue of the fundamental transfer matrix is
\begin{equation}
  t(u)=\sum_{j=1}^{N} z^{(j)}(u),
\end{equation}
where
\begin{equation}
  \label{zdef}
 z^{(\ell)} (u) = \frac{Q_0(u)}{Q_0(u+\ii)} \frac{Q_{\ell-1}(u+\ii\frac{\ell+1}2)Q_{\ell}(u+\ii\frac{\ell-2}2)}{Q_{\ell-1}(u+\ii\frac{\ell-1}2)Q_{\ell}(u+\ii\frac{\ell}{2})}, \quad \ell =1, \dots ,N.
\end{equation}

The Bethe equations are
\begin{equation}
  \begin{split}
\frac{Q_{\ell-1}(\lambda_j^{(\ell)}+\ii\frac{1}{2})}{Q_{\ell-1}(\lambda_j^{(\ell)}-\ii\frac{1}{2})}
\frac{Q_{\ell}(\lambda_j^{(\ell)}-\ii)}{Q_{\ell}(\lambda_j^{(\ell)}+\ii)}
&\frac{Q_{\ell+1}(\lambda_j^{(\ell)}+\ii\frac{1}{2})}{Q_{\ell+1}(\lambda_j^{(\ell)}-\ii\frac{1}{2})}=-1,\\
&\qquad j=1\dots N_\ell,\quad
\ell=2,\dots,N-1.
  \end{split}
\end{equation}
With our conventions the 1-string solutions to each nesting level are
real rapidities.

Fused transfer matrices $t^{\Lambda}(u)$ can be constructed for every
irreducible representation $\Lambda$ of $SU(N)$
\cite{kulish-resh-sklyanin--fusion,bazhanov-reshetikhin-rsos-fusion}, but a
special role is 
played by those representations that correspond to rectangular Young
diagrams. For the diagram with $a$ rows and $m$ columns the
corresponding transfer matrix is denoted by $t^{(a)}_m(u)$. These
objects satisfy the Hirota equation
\begin{equation}
  \label{suNt}
  \begin{split}
  t^{(a)}_m\left(u+\frac\ii2\right)t^{(a)}_m\left(u-\frac\ii2\right)&=t^{(a)}_{m+1}(u)t^{(a)}_{m-1}(u)+
  t^{(a-1)}_{m}(u)t^{(a+1)}_{m}(u),\\
 &\hspace{2cm}  a=1,\dots,N-1\qquad m=1,2,\dots.
  \end{split}
 \end{equation}
We specify the boundary conditions to this system 
motivated by the $SU(3)$ case: we require
\begin{equation}
  \label{suNfura}
t^{(a)}_0=1,\qquad  t^{(0)}_m(u)=\frac{Q_0(u-\ii \frac{m}{2})}{Q_0(u+\ii \frac{m}{2})},
  \qquad t^{(N)}_m(u)=1.
\end{equation}
Together with \eqref{suNt} this completely determines the
normalization of the local Lax operators and the fused transfer
matrices. 

For each $t^{(a)}_m(u)$ we define its space reflected variant $\bar
t^{(a)}_m(u)$. By using local crossing relations and the conjugation
properties of $SU(N)$ representations we can express each $\bar t_m^{(a)}(u)$
 using $t^{(N-a)}_m(u)$; the resulting relations are generalizations
 of \eqref{tmcrossing} and \eqref{su3crossing}. The space reflected
 TM's are thus not independent, but the precise relation is irrelevant
 for our purposes.

Based on the previous results we formulate:
\begin{conj}
   For $u\in\physstr$ the following asymptotic inversion relations hold:
\begin{equation}
  \label{tmNi}
\begin{split}
  \bar{t}_m^{(a)} (-u) t_m^{(a)} (u)  \approx 1.
\end{split} 
\end{equation}
\end{conj}

In analogy with the earlier results we introduce the operators
\begin{equation}
\begin{split}
\label{eq:tmkDef}
X^{(a)}_m (u) = (-\ii)\partial_u \log t^{(a)}_m(u)\approx
 (-\ii)\bar t^{(a)}_m(-u)\partial_u t^{(a)}_m(u).
\end{split}
\end{equation}

\begin{conj}
For $u\in\physstr$ the traceless operators $\{X_m^{(a)}(u)\}$ are quasi-local.
\end{conj}

The string-charge relations can be established if the eigenvalues of
the fused transfer matrices are known. General expressions using Young
tableaux were derived in \cite{bazhanov-reshetikhin-rsos-fusion}, see also
Section 7 of the review
\cite{suzuki-kuniba-tomoki-t-system-y-system-review}.
The rule to compute the eigenvalues is the following. Let us consider the
    $(a\times m)$ Young diagram, and all possible semi-standard Young
    tableaux, i.e. the filling of the diagram with numbers $1,\dots,N$
    such that they are increasing from top to bottom and
    non-decreasing from left to right. For example, for the Young
    diagram
    \begin{equation}
      \yng(2,2)  
    \end{equation}
    the possible semi-standard tableaux for $N=3$ are
\begin{equation}
   \young(11,22),\
   \young(11,23),\
   \young(11,33),\
   \young(12,23),\
   \young(12,33),\
   \young(22,33)
\end{equation}
In the following let $\tau_{kl}$ denote the element  of a tableau $\tau$ in
row $k=1\dots a$ and column $l=1\dots m$ from the top left. Then the
formula for the eigenvalues is \cite{bazhanov-reshetikhin-rsos-fusion}
\begin{equation}
  \label{tam}
    t^{(a)}_m(u)=\prod_{j=1}^{a-1} \frac{Q_0(u+\ii\frac{m-a+2j}{2})}{Q_0(u-\ii\frac{m-a+2j}{2})}
\times    \sum_{\tau }
\left[ \mathop{\prod_{k=1\dots a}}_{l=1\dots m} z^{(\tau_{kl})}\left(u+\ii \frac{a-m-2k+2l}{2}\right)\right].
  \end{equation}
Here the sum runs over all allowed semi-standard tableaux of size
$(a\times m)$ for the
given $N$, and the $z$-functions are defined in \eqref{zdef}. The presence of the
pre-factor before the sum is a consequence of our
normalization.

In a perhaps more direct way the transfer matrices can be expressed
as \cite{bazhanov-reshetikhin-rsos-fusion} 
\begin{align}
 t^{(a)}_m (u) &= \det \left( t^{(1)}_{m-i+j}\big( u+\ii\frac{i+j-1-a}{2} \big) \right)_{1\le i,j\le a} \\
 \label{eq:Tsystemdetsol}
 &= \det \left( t^{(a-i+j)}_1 \big( u+\ii\frac{m-i-j+1}{2} \big) \right)_{1\le i,j \le m},
\end{align}
using the two series $t^{(1)}_m$ or $t_{1}^{(a)}$, for which we have
the formulas
  \begin{equation}
    \begin{split}
  \label{ta1}
  t_m^{(1)} (u) &=
      \sum_{1\le i_1 \le i_2 \le \dots \le i_m \le N}
\left[\prod_{\ell=1}^m z^{(i_\ell)}\left(u+\ii \frac{-1-m+2\ell}{2}\right)\right]\\
  t_1^{(a)} (u) &=
\frac{Q_0(u+\frac{a-1}{2}\ii)}{Q_0(u-\frac{a-1}{2}\ii)}\times
         \sum_{1\le i_1 < i_2 < \dots < i_a \le N}\left[
\prod_{\ell=1}^a z^{(i_\ell)}\left(u+\ii \frac{a+1-2\ell}{2}\right)\right].
\end{split}
  \end{equation}
  The latter formulas are special cases of the general
  tableaux sum.

  \begin{thm}
    For each $a=1\dots N-1$ and $m=1,2,\dots$ there is a single term
    in the expansion \eqref{tam} of $t^{(a)}_m(u)$ which automatically produces
    the asymptotic inversion \eqref{tmNi}. This term corresponds to
    the Young tableau where all elements of row $k$ are equal to $k$
    for $k=1,\dots,a$. The explicit form of this term is
    \begin{equation}
      \label{Qam}
     \frac{Q_a\left(u-\ii\frac m2\right)}{Q_a\left(u+\ii\frac m2\right)}.
    \end{equation}
  \end{thm}
  \begin{proof}
This can be proven recursively: starting from the bottom right element
of the Young diagram, and afterwards considering the elements in the
upper rows. If $\tau_{am}$ is the element in the bottom right corner,
then this number can only be present in the bottom row of the
diagram. The corresponding $z$-factors will involve the $Q$-functions
$Q_0$, $Q_{\tau_{am}-1}$ and $Q_{\tau_{am}}$. The factors of
$Q_{\tau_{am}}$ in the product can only come from the boxes filled
with $\tau_{am}$, which can occupy a number of cells to the left of
the bottom right corner. Collecting these factors we see that the resulting
combination of the $Q_{\tau_{am}}$ functions can yield a form
satisfying the inversion relation only if $\tau_{am}=a$ and this
number occupies the full bottom row; this is a simple consequence of the
various shifts present in the $z$-functions and the $Q$-functions. The
proof continues by considering the element $\tau_{a-1,m}$ and the
corresponding $Q$-functions, showing that the only possibility is that
$\tau_{a-1,m}=a-1$ and this number has to fills the row $a-1$. This is then
repeated for all rows upwards. Collecting all factors of the $Q_0$
functions we see that they just cancel each other, and we obtain \eqref{Qam}.
  \end{proof}

Using our arguments of the previous Sections it follows that in the
thermodynamic limit this term has to be dominant for $u\in\physstr$:
\begin{equation}
  \label{tmalead}
 t_m^{(a)} (u) \approx \frac{Q_a\left(u-\ii\frac m2\right)}{Q_a\left(u+\ii\frac m2\right)},
\qquad a=1,\dots,N-1,\quad m=1,2,\dots
\end{equation}
A special case of this statement already appeared in a closely related
problem in \cite{bazhanov-reshetikhin-rsos-fusion}.  Note that due to our
boundary conditions \eqref{suNfura} the above relation holds even for
$a=0$ and $a=N$, using the convention $Q_N(u)=1$.

Based on \eqref{tmalead} and the similarities in the derivation of the TBA
equations for all $N$, we propose the following general pattern:
\begin{conj}
In the $SU(N)$-symmetric model the string-charge relations are
\begin{equation}
\begin{split}
 \rho_m^{(a)} &= \frac{1}{2\pi L}\left(X_m^{(a)[+]} + X_m^{(a)[-]} - X^{(a)}_{m+1} - X^{(a)}_{m-1}\right), \\
&\hspace{2cm}  a=1\dots N-1, \qquad m=1\dots\infty,
\end{split}
\end{equation}
where $\rho^{(a)}_m(u)$ is the density of $m$-strings of rapidity type $a$.
\end{conj}

\section{Conclusions}

\label{sec:conclusions}

We studied the GET (Generalized
Eigenstate Thermalization) for higher rank spin models with
$SU(N),\ N\ge 3$ symmetry, with the main focus being on the $N=3$ Lai-Sutherland
model. We argued that a complete set of charges is obtained from the
known fusion hierarchy of transfer matrices.
These fused transfer matrices correspond to the representations of
$SU(3)$ with rectangular diagrams of size $1\times m$ and
$2\times m$, $m=1,\dots,\infty$, or equivalently, to symmetrically
fused defining and conjugate representations, respectively.
We computed the
thermodynamic limit of these charges: the resulting string-charge
relations take essentially the same form as in the $SU(2)$-invariant
XXX chain, with the simple extension of having two particle types and
two series of fused charges.

These results in the $SU(3)$ chain possess a conjugation symmetry:
exchanging the two families of fused transfer matrices corresponds to
exchanging the defining and conjugate representations of $SU(3)$. On
the level of the string-charge relations this is reflected by an
exchange of the two particle types. The final relations
\eqref{eq:SU3SCdualityV3} are completely symmetric with respect to
this conjugation.

The strictly local charges of the model are
only sensitive to the particles of the first type, and their finite
volume mean values are computed as
\begin{equation}
  \mathcal{Q}_\alpha=\sum_{j=1}^{N_1} q_\alpha(\lambda^{(1)}_j),
\end{equation}
with $q_\alpha(u)$ being the one-particle eigenvalues.
This raises the
question: are there local or quasi-local operators which are sensitive
only to the second type of particles? Our results show that the
$Y_m(u)$ are such operators {\it in the thermodynamic limit, for almost
all Bethe states.} There is no local operator whose
finite volume eigenvalues would take the form
\begin{equation}
  Y=\sum_{j=1}^{N_2}  f(\lambda^{(2)}_j).
\end{equation}
The string-charge relations are only found in the thermodynamic limit,
and hold only for {\it almost all Bethe states}. The crucial step is 
 the dominance of a prescribed term in the expressions of the transfer
 matrix eigenvalues. 

Based on the derivation for $N=2,3$
we conjectured generic results for arbitrary $N$. In particular, we
conjectured that the complete GGE is formed by the 
charges built on the transfer matrices corresponding to the $(a\times
m)$ Young diagrams with $a=1\dots N-1$ and $m=1\dots\infty$.

In some sense these results are not surprising. It is known that in an
integrable model with symmetry group $G$ the structure of the
nested Bethe Ansatz (for example, the set of the Bethe equations)
mirrors  the Dynkin diagram of $G$. The fusion rules for the transfer
matrices also closely follow the
Dynkin diagram. It is thus not surprising that the $T$-system and the
set of TBA equations determining the Bethe root densities are so
closely related. This correspondence was noted and used in a large
number of works already at the end of the 80's and beginning of the
90's; for concrete references see the thorough review
\cite{suzuki-kuniba-tomoki-t-system-y-system-review}. 
The new addition to the theory was the discovery that in the XXX model
the formal expressions of the type $\partial_u \log(t_m(u))$ yield quasi-local
operators for $u\in\physstr$
\cite{prosen-xxx-quasi,prosen-enej-quasi-local-review}, and that in
the TDL they contain just enough  information to fix all root
densities. What we have performed in this work is to extend this
observation to the $SU(N)$-symmetric fundamental models.

In accordance, the crucial points of our work are the proofs of the
inversion relations and the quasi-locality. We argued that the
quasi-locality property follows once the inversion relation is established on
the  level of the operators. We computed a detailed proof
in two cases, namely for the operators $t^{(1)}_1(u)$ and
$t^{(2)}_1(u)$ for $SU(3)$, which correspond to the defining and
conjugate representations. Although these are just two particular cases, we
believe they constitute strong justification for the remaining
conjectures. Also, we remind that even
in the $SU(2)$ case complete analytical proofs are available only up
to $s=3/2$ \cite{prosen-xxx-quasi,prosen-enej-quasi-local-review}.

Naturally, it would be desirable to have explicit proofs in
more cases, possibly for the whole fusion hierarchy. On
the technical level, the task to be performed is the diagonalization
of a 4-site transfer matrix, where these 4 sites carry some
fused representations of the symmetry group. Such transfer matrices
can be diagonalized by the Bethe Ansatz, and a quite general and completely analytical
approach is detailed for example in \cite{fast-bethe-solver}. It
remains to be seen whether this or any alternative techniques
prove to be useful for the problem at hand.

Also, it would be interesting to find a more general prescription for the
GGE. Based on our results it seems plausible
 that the string-charge relations are always encoded in the
known fusion hierarchy of the theory
\cite{junji-suzuki-kuniba-tomoki-fusion,suzuki-kuniba-tomoki-t-system-y-system-review}.

\vspace{1cm}
{\bf Acknowledgments} 
\bigskip

The authors would like to thank Tam\'as Gombor, Enej Ilievski, M\'arton Mesty\'an,
Junji Suzuki, and G\'abor Tak\'acs for useful discussions.
This
research was supported by the BME-Nanotechnology FIKP grant
of EMMI (BME FIKP-NAT), by the National Research Development
and Innovation Office (NKFIH) (K-2016 grant no. 119204, the OTKA
grant no. SNN118028, and the KH-17 grant no. 125567), and by the ``Premium'' Postdoctoral
Program of the Hungarian Academy of Sciences.

\bigskip

\appendix

\section{Derivation of the asymptotic inversion and the quasi-locality
property}

\label{sec:proofs}

We present here the derivation of the asymptotic inversion and
quasi-locality property for $t(u)$ and $t^{\bar 3}(u)$. The two
computations are quite similar, using same techniques and
the same objects.

We need the adjoints of transfer matrices, hence we start with expressing them.\\
\begin{lemma}
  \label{adjointlem}
 The adjoints of the $t,t^{\bar 3}$ transfer matrices, and their respective space reflected pairs are the following:
 \begin{align}
  t^\dagger(u) &= \bar t(-u^*) & \big(t^{\bar 3}\big)^\dagger(u)&= \bar{t}^{\bar 3} (-u^*)\\
  \bar{t}^\dagger(u) &= t(-u^*) & \big(\bar{t}^{\bar 3}\big)^\dagger(u)&= t^{\bar 3} (-u^*).
 \end{align}
\end{lemma}
\begin{proof}
Denote by $A^*$ the complex conjugate of $A$, for any
matrix or complex number. If $A$ is matrix, the complex conjugation is
considered element-wise.

Direct computation shows that the $R$-matrices satisfy:
\begin{equation}
 \left(R(u)\right)^* = R(-u^*),\qquad \left(R^{\bar 3}(u)\right)^* = R^{\bar 3}(-u^*).
\end{equation}
Taking transposition before partial trace reverses the order of $R$-matrices, from where the statement follows.   
\end{proof}

\textbf{Proof of asymptotic inversion.} To prove the \eqref{eq:inv1} asymptotic inversion, we prove the following two statements regarding the norms of the operators:
\begin{equation}\label{eq:AsympInvNorm}
\begin{split}
 \|t(u)\bar t(-u) -1\|^2_{\text{HS}} &\approx 0 \\
 \|t^{\bar 3}(u)\bar t^{\bar 3}(-u) -1 \|^2_{\text{HS}} &\approx 0.
\end{split}
\end{equation}
Expanding the l.h.s. of these expressions leads to the following:
\begin{equation}
\begin{split}
 \|t(u)\bar t(-u) -1\|^2_{\text{HS}} &= 3^{-L} \Tr\, t(u^*)\bar t(-u^*) t(u)\bar t(-u) - 3^{-L} 2\, \Re\, \Tr\, t(u)\bar t(-u) +1\\
 \|t^{\bar 3}(u)\bar t^{\bar 3}(-u) -1 \|^2_{\text{HS}} &= 3^{-L} \Tr\, t^{\bar 3}(u^*)\bar t^{\bar 3}(-u^*) t^{\bar 3}(u)\bar t^{\bar 3}(-u) - \\ & \qquad\qquad\qquad\qquad\qquad-3^{-L} 2\, \Re\, \Tr\, t^{\bar 3}(u)\bar t^{\bar 3}(-u) +1.
\end{split}
\end{equation}
As explained in Section \ref{sec:invqloc} the traces can be expressed
as a partition function of a 2D lattice model, which can be
alternatively evaluated by the Quantum Transfer Matrices acting in the
``crossed channel''. See also \cite{prosen-enej-quasi-local-review}.

We thus get
\begin{equation}
  \label{ujabbtraces}
\begin{split}
 \|t(u)\bar t(-u) -1\|^2_{\text{HS}} &= 3^{-L} \Tr\, t_{AB} (u,u,u^*,u^*)^L -3^{-L}2\, \Re\, \Tr\, t_A(u,u)^L +1, \\
 \|t^{\bar 3}(u)\bar t^{\bar 3}(-u) -1 \|^2_{\text{HS}} &= 3^{-L} \Tr\, \bar{t}_{AB} (u,u,u^*,u^*)^L -3^{-L}2\, \Re\, \Tr\, \bar{t}_A(u,u)^L +1,
\end{split}
\end{equation}
where 
\begin{equation}
\begin{split}
 t_{AB}(u_1,u_2,v_1,v_2)&= \Tr_a R_{a,1}(v_2) R_{a,2}^t(-v_1) R_{a,1}(u_1) R_{a,2}^t(-u_2) \\
 t_A(u_1,u_2)&=\Tr_a R_{a,2}^t(-u_2)R_{a,1}(u_1)\\
 t_{AB}^{\bar 3}(u_1,u_2,v_1,v_2)&= \Tr_a R_{a,1}^{\bar 3}(v_2) (R_{a,2}^{\bar 3})^t(-v_1) R_{a,1}^{\bar 3}(u_1) (R_{a,2}^{\bar 3})^t(-u_2) \\
 t_A^{\bar 3}(u_1,u_2)&=\Tr_a (R_{a,2}^{\bar 3})^t(-u_2)R_{a,1}^{\bar 3}(u_1).
\end{split}
\end{equation}
The traces in \eqref{ujabbtraces} can be expressed using the
eigenvalues of the above matrices. Let us denote by
$\lambda_{AB,j}(u_1,u_2,v_1,v_2),\ j =1\dots 3^4$ and
$\lambda_{A,j}(u_1,u_2),\ j =1\dots 3^2$ the eigenvalues of
$t_{AB}(u_1,u_2,v_1,v_2)$ and $t_A(u_1,u_2)$, respectively. Similarly,
we denote by $\bar\lambda_{AB,j}(u_1,u_2,v_1,v_2),\ j =1\dots 3^4$ and
$\bar\lambda_{A,j}(u_1,u_2),\ j =1\dots 3^2$ the eigenvalues for
$\bar{t}_{AB}(u_1,u_2,v_1,v_2)$ and $\bar{t}_A(u_1,u_2)$,
respectively.

Then the squared norms are the following:
\begin{equation}
  \label{qwe}
\begin{split}
 \|t(u)\bar t(-u) -1\|^2_{\text{HS}} &= 3^{-L} \sum_{j=1}^{3^4} \lambda_{AB,j}^L(u,u,u^*,u^*) -3^{-L}2\, \sum_{j=1}^{3^2} \lambda_{A,j}^L(u,u) +1 \\
 \|t^{\bar 3}(u)\bar t^{\bar 3}(-u) -1 \|^2_{\text{HS}} &= 3^{-L} \sum_{j=1}^{3^4} \bar\lambda_{AB,j}^L(u,u,u^*,u^*) -3^{-L}2\, \sum_{j=1}^{3^2} \lambda_{A,j}^L(u,u) +1.
\end{split}
\end{equation}
As explained in Sec. \ref{sec:invqloc}, the local inversion relations
imply that the delta-states are eigenvectors of the QTM's with trivial
eigenvalues equal to 3. If the eigenvalue 3  is non-degenerate and
dominant for all of these matrices, then \eqref{qwe} implies
\eqref{eq:AsympInvNorm}.
Hence, what remains is to prove that $3$ is indeed a leading,
non-degenerate eigenvalue for all of these matrices, as long as $u\in\physstr$. 

This will be proven somewhat later in this Section. First we consider
the quasi-locality property, because its proof also involves the same
QTM construction.
\bigskip

\textbf{Proof of quasi-locality.}
Starting from the definitions
\begin{equation}
\begin{split}
 X(u)&=(-\ii) \bar t(-u)\partial t(u)\\
 Y(u)&=(-\ii) \bar t^{\bar 3}(-u)  \partial t^{\bar 3}(u)
\end{split}
\end{equation}
we compute the adjoints of the generator functions:
\begin{equation}
\begin{split}
 X^\dagger (u) &= (-\ii) \left.\left(\partial\bar t(v)t(-v)\right)\right|_{v=-u^*} = (-\ii)\left.\left(t(-v)\partial \bar t(v) \right)\right|_{v=-u^*} \\
 Y^\dagger (u) &= (-\ii) \left.\left(\partial\bar{t}^{\bar 3}(v)t^{\bar 3}(-v)\right)\right|_{v=-u^*} = (-\ii)\left.\left(t^{\bar 3}(-v)\partial \bar{t}^{\bar 3}(v) \right)\right|_{v=-u^*}.
\end{split}
\end{equation}
Here we used the commutativity of the transfer matrix with it derivative. Using
\begin{equation}
 \|\{A\}\|_{\text{HS}}^2 = 3^{-L}\Tr\left(A^\dagger A\right)-3^{-2L}\Tr\left(A^\dagger \right)\Tr\big(A\big)
\end{equation}
we get
\begin{equation}
  \begin{split}
    \|X(u)\|_{\text{HS}}^2=-3^{-L} \Tr\,{ t(u^*)\partial \bar t(-u^*)  \bar  t(-u)\partial t(u)}
    +3^{-2L} \Tr\,{t(u^*)\partial \bar t(-u^*)}
     \Tr\,{\bar t(-u)\partial t(u)} \\
    \|Y(u)\|_{\text{HS}}^2=-3^{-L} \Tr\,{ t(u^*)\partial \bar t(-u^*)  \bar  t(-u)\partial t(u)}
    +3^{-2L} \Tr\,{t(u^*)\partial \bar t(-u^*)}
     \Tr\,{\bar t(-u)\partial t(u)}.
  \end{split}
\end{equation}
We consider these traces once more in the rotated channel. In fact, we
consider them as two special points of the more general expressions
\begin{equation}
\begin{split}
 \mathcal{K}(u_1,u_2,v_1,v_2) &= \partial_{v_1}\partial_{u_1}\left[ 3^{-L}\Tr\,{ t(v_2)\bar t(-v_1)  \bar  t(-u_2) t(u_1)} \right. - \\ &\qquad\qquad\qquad\qquad\qquad\left.- 3^{-2L} \Tr\,{t(v_2) \bar t(-v_1)} \Tr\,{\bar t(-u_2) t(u_1)}\right], \\
 \bar{\mathcal{K}}(u_1,u_2,v_1,v_2) &= \partial_{v_1}\partial_{u_1}\left[
  3^{-L} \Tr\,{ t^{\bar 3}(v_2)\bar t^{\bar 3}(-v_1)  \bar  t^{\bar 3}(-u_2) t^{\bar 3}(u_1)} -\right. \\ &\qquad\qquad\qquad\qquad\qquad\left. -3^{-2L}\Tr\,{t^{\bar 3}(v_2) \bar t^{\bar 3}(-v_1)}
     \Tr\,{\bar t^{\bar 3}(-u_2) t^{\bar 3}(u_1)}\right], 
 \end{split}
\end{equation}
such that
\begin{equation}
      \| X(u) \|_{\text{HS}}^2 = \mathcal K (u,u,u^*,u^*), \qquad\qquad\qquad
 \|Y(u)\|_{\text{HS}}^2 = \bar{\mathcal{K}}(u,u,u^*,u^*). 
\end{equation}
We compute $\mathcal K$ and $\bar{\mathcal K}$ in the rotated channel,
similarly as before in the proof of the asymptotic inversion
relations. Using the QTM's the quantities $\mathcal K$ and $\bar{\mathcal K}$ can
be expressed as
\begin{equation}
\begin{split}
   \mathcal{K}(u_1,u_2,v_1,v_2)&=\partial_{v_1}\partial_{u_1}\left[ 3^{-L} \Tr\,\left( t_{AB}(u_1,u_2,v_1,v_2) \right)^L -\right. \\ &\qquad\qquad\qquad\qquad\qquad\left.- 3^{-2L} \Tr\,\left( t_A(u_1,u_2) \right)^L  \Tr\,\left( t_A(v_1,v_2)\right)^L \right] \\
   \bar{\mathcal{K}}(u_1,u_2,v_1,v_2)&=\partial_{v_1}\partial_{u_1}\left[ 3^{-L} \Tr\,\left( t_{AB}^{\bar 3}(u_1,u_2,v_1,v_2) \right)^L -\right. \\ &\qquad\qquad\qquad\qquad\qquad\left. 3^{-2L} \Tr\,\left( t_A^{\bar 3}(u_1,u_2) \right)^L  \Tr\,\left( t_A^{\bar 3}(v_1,v_2)\right)^L \right].
\end{split}
\end{equation}
As before, we consider these expressions in the eigenbases of the respective matrices:
\begin{equation}
\begin{split}
 &\mathcal{K}(u_1,u_2,v_1,v_2) = 3^{-L}\sum_{j=1}^{3^4} L \lambda_{AB,j}^{L-1}  (\partial_{v_1} \partial_{u_1} \lambda_{AB,j}) + L(L-1)\lambda_{AB,j}^{L-2}(\partial_{v_1}\lambda_{AB,j})\times\\ &\qquad\qquad\qquad\qquad\times(\partial_{u_1}\lambda_{AB,j})
 - 3^{-2L} \sum_{j=1}^{3^2} L \lambda_{A,j}^{L-1}\partial_{u_1} \lambda_{A,j} \sum_{j=1}^{3^2} L \lambda_{A,j}^{L-1}\partial_{v_1} \lambda_{A,j} \\
 &\bar{\mathcal{K}}(u_1,u_2,v_1,v_2) = 3^{-L}\sum_{j=1}^{3^4} L \bar{\lambda}_{AB,j}^{L-1}  (\partial_{v_1} \partial_{u_1} \bar{\lambda}_{AB,j}) + L(L-1)\bar{\lambda}_{AB,j}^{L-2}(\partial_{v_1}\bar{\lambda}_{AB,j})\times\\ &\qquad\qquad\qquad\qquad\times(\partial_{u_1}\bar{\lambda}_{AB,j}) -  3^{-2L} \sum_{j=1}^{3^2} L \bar{\lambda}_{A,j}^{L-1}\partial_{u_1}  \bar{\lambda}_{A,j} \sum_{j=1}^{3^2} L \bar{\lambda}_{A,j}^{L-1}\partial_{v_1} \bar{\lambda}_{A,j}.
\end{split}
\end{equation}
In the notations above we suppressed the dependence of the eigenvalues on the spectral
parameters. For the partial derivatives of $\lambda_{A,j},\ \bar\lambda_{A,j}$, the following arguments are understood: 
\begin{align}
\partial_{u_1}\lambda_{A,j} &\equiv \partial_{u_1}\lambda_{A,j}(u_1,u_2) & \partial_{v_1}\lambda_{A,j} &\equiv \partial_{v_1}\lambda_{A,j}(v_1,v_2) \\
\partial_{u_1}\bar\lambda_{A,j} &\equiv \partial_{u_1}\bar\lambda_{A,j}(u_1,u_2) & \partial_{v_1}\bar\lambda_{A,j} &\equiv \partial_{v_1}\bar\lambda_{A,j}(v_1,v_2)
\end{align}
Let us assume that for each matrix above there is a non-degenerate dominant eigenvalue with
index $j=1$. Then the leading terms in the above sums are
\begin{equation}
\begin{split}
 &\mathcal{K}(u_1,u_2,v_1,v_2) \approx L^2 \left(3^{-L}\lambda_{AB,1}^{L-2} \partial_{u_1}\lambda_{AB,1}\partial_{v_1}\lambda_{AB,1} - 3^{-2L} \lambda_{A,1}^{L-1}\partial_{u_1}\lambda_{A,1} \lambda_{A,1}^{L-1}\partial_{v_1}\lambda_{A,1} \right) + \\& +L \left(3^{-L} \lambda_{AB,1}^{L-1}\partial_{v_1}\partial_{u_1}\lambda_{AB,1} - 3^{-L}\lambda_{AB,1}^{L-2}\partial_{v_1}\lambda_{AB,1}\partial_{u_1}\lambda_{AB,1}\right) \\
  &\bar{\mathcal{K}}(u_1,u_2,v_1,v_2) \approx L^2 \left(3^{-L}\bar{\lambda}_{AB,1}^{L-2} \partial_{u_1}\bar{\lambda}_{AB,1}\partial_{v_1}\bar{\lambda}_{AB,1} - 3^{-2L} \bar{\lambda}_{A,1}^{L-1}\partial_{u_1}\bar{\lambda}_{A,1} \bar{\lambda}_{A,1}^{L-1}\partial_{v_1}\bar{\lambda}_{A,1} \right) + \\& +L \left(3^{-L} \bar{\lambda}_{AB,1}^{L-1}\partial_{v_1}\partial_{u_1}\bar{\lambda}_{AB,1} - 3^{-L}\bar{\lambda}_{AB,1}^{L-2}\partial_{v_1}\bar{\lambda}_{AB,1}\partial_{u_1}\bar{\lambda}_{AB,1}\right).
\end{split}
\end{equation}
The operators $\{X(u)\}$ and $\{Y(u)\}$ are quasi-local if the norm is
of $\ordo(L)$, i.e. the $\ordo(L^2)$ terms cancel. This will be
investigated at the point $u_1=u_2$, $v_1=v_2$, which has to be
substituted after taking the partial derivatives. As we will see, the
eigenvalue $3$ is leading and non-degenerate at  $u_1=u_2=u$,
$v_1=v_2=u^*$, $u\in\physstr$, therefore we need to prove that the relations
\begin{equation}
  \begin{split}
  \partial_{u_1}\bar\lambda_{AB,1}\partial_{v_1}\bar\lambda_{AB,1}
  - \partial_{u_1}\bar\lambda_{A,1} \partial_{v_1}\bar\lambda_{A,1}  =0    \\
  \partial_{u_1}\bar\lambda_{AB,1}\partial_{v_1}\bar\lambda_{AB,1}
  - \partial_{u_1}\bar\lambda_{A,1} \partial_{v_1}\bar\lambda_{A,1}  =0
\end{split}
\end{equation}
hold at  $u_1=u_2$, $v_1=v_2$. As explained in
\cite{prosen-enej-quasi-local-review}, this follows from the
factorizability of the leading eigenvector and the 
Hellmann-Feynman theorem.

We thus obtain the final result 
\begin{equation}
  \begin{split} 
  \mathcal{K}(u_1,u_2,v_1,v_2) &\approx L
\left( \frac{1}{3}\partial_{u_1}\partial_{v_1}\lambda_{AB,1}
-
\frac{1}{9}\partial_{u_1}\lambda_{AB,1}  \partial_{v_1}\lambda_{AB,1}
\right)\\
  \bar{\mathcal{K}}(u_1,u_2,v_1,v_2) &\approx L
\left( \frac{1}{3}\partial_{u_1}\partial_{v_1}\bar{\lambda}_{AB,1}
-
\frac{1}{9}\partial_{u_1}\bar{\lambda}_{AB,1}  \partial_{v_1}\bar{\lambda}_{AB,1}
\right).
  \end{split}
\end{equation}

What remains to be proven is that the delta-states with eigenvalues
$3$ are indeed dominant for $t_{AB}$, $\bar t_{AB}$, $t_A$ and $\bar t_{A}$ and
$u\in\physstr$.
\subsection{The eigensystem of $t_A$ and $\bar{t}_{A}$}
Here we determine the eigenvalues of $t_A$ and $\bar{t}_{A}$ by a
direct computation.

Let us denote the auxiliary spaces of the original physical transfer
matrices $t(u)$ and $\bar t(u)$ by $0$ and $\bar 0$, respectively.
The monodromy matrix in the crossed channel is 
\begin{equation}
  M_A(u)= R_{A,\bar 0}^{t_{\bar 0}}(-u)R_{A,0}(u).
\end{equation}
Here $A$ stands for the auxiliary space of the rotated transfer
matrix, which is identical to the original physical quantum space. 

It follows simply from the unitarity relation \eqref{eq:uni1} that the
delta-state
is an eigenvector of $t_A(u)=\text{Tr}_A M_A(u)$ with
eigenvalue 3. Now we compute the full spectrum directly. Writing out the operators
and taking the trace in $A$ we get
\begin{equation}
  t_A(u)=\frac{-u^2 3-K_{0\bar 0}}{-1-u^2}.
\end{equation}
The eigenvalues of $K$ are $3$ for the trace vector and 0 otherwise, so
the eigenvalues are thus
\begin{equation}
  3,\qquad 3\frac{u^2}{1+u^2}.
\end{equation}
The trace vector is the leading eigenvector whenever
$|u^2/(1+u^2)|<1$. This is true until $\Re(u^2)>-1/2$. Writing
$u=a+\ii b$ the condition is
\begin{equation}
  \label{c1}
  a^2-b^2>-1/2.
\end{equation}

In the second case, we need to construct the monodromy matrix
\begin{equation}
  \bar M_A(u)= \bar R_{A,0}^{t_0}(-u)\bar R_{A,\bar 0}(u)=
  \frac{\left(-u+\frac\ii2\right)\left(u+\frac\ii2\right)}{\left(-u+\frac{3\ii}2\right)\left(u+\frac{3\ii}2\right)}
   R_{A,0}(u-\sigma) R^t_{A,\bar 0}(-u-\sigma).
\end{equation}
More explicitly it reads
\begin{equation}
  \bar M_A(u)=
\frac{\left(u-\frac\ii2 + P_{A,0}\right)\left(-u-\frac{3\ii}2+\ii K_{A,\bar 0}\right)}
{\left(-u+\frac{3\ii}2\right)\left(u+\frac{3\ii}2\right)}.
\end{equation}
The trace becomes
\begin{equation}
  t_B(u)=\frac{3\left(\left(\frac{3\ii}2\right)^2-u^2\right)+3-K_{0\bar 0} }
{\left(\frac{3\ii}2\right)^2-u^2}.
\end{equation}
Hence the eigenvalues are
\begin{equation}
  3,\qquad  \qquad 3\frac{\frac{3^2}{4}+u^2-1}{\frac{3^2}{4}+u^2}.
\end{equation}
The first one, corresponding to the delta state, is the leading eigenvalue if
\begin{equation}
  \label{c2}
  \Re(u^2) > 1/2-3^2/4=-\frac{7}{4}.
\end{equation}
Out of the two conditions \eqref{c1} and \eqref{c2}, the first one is
more restricting, and it holds in the physical strip.
\subsection{The eigensystem of $t_{AB}$ and $\bar{t}_{AB}$}
Here we compute the spectrum of the 4-site QTM's $t_{AB}(u_1,u_2,v_1,v_2)$ and
$\bar{t}_{AB}(u_1,u_2,v_1,v_2)$ at $u_1=u_2=u,\ v_1=v_2=v$. Both matrices correspond to inhomogeneous, integrable
4-site spin chains, therefore their spectrum can be determined by
Bethe Ansatz techniques. One possibility would be the application of
the $T-Q$ relations, in the spirit of
\cite{fast-bethe-solver}. However, these are still relatively small
matrices of size $81\times 81$, therefore a direct method is perhaps faster. An exact
diagonalization of the matrices is not possible, because the entries
are functions of two parameters $u,v$, and the diagonalization
lies beyond the capabilities of the symbolic manipulation programs
such as \texttt{Mathematica}. Nevertheless, the direct diagonalization is
possible after simplifying the matrix using some group theory arguments.

Both transfer matrices are $SU(3)$-invariant, therefore their spectrum
can be analyzed by constructing the Clebsch-Gordan series. 
The transfer matrix $t_{AB}$ acts on the tensor product space
\begin{equation}
  \label{haha}
 3\otimes \bar 3\otimes 3\otimes \bar 3,
\end{equation}
whereas $\bar t_{AB}$ acts on
\begin{equation}
  \bar3\otimes3\otimes\bar3\otimes3,
\end{equation}
where we used the standard notations for the defining and conjugate
representations. In terms of Young diagrams they are denoted as
\begin{equation}
  3=\yng(1),\qquad \bar 3=\yng(1,1).
\end{equation}
In the following we perform a simple permutation for the vector spaces
of $\bar t_{AB}$ such that it also acts on the space given by \eqref{haha}.

The Clebsch-Gordan series can be computed with standard methods. We
obtain 
\begin{equation}
  \yng(1)\otimes \yng(1,1)\otimes
    \yng(1)\otimes \yng(1,1) = \yng(4,2)\ \oplus\ \yng(3)\ \oplus\ \yng(3,3)\ \oplus\ 4\cdot
  \yng(2,1)\ \oplus\ 2 \cdot 1.
\end{equation}
An easy check-back on the dimensions is 
\begin{equation}
  3 \cdot 3 \cdot 3 \cdot 3 = 27+10+10+4\cdot 8+2\cdot 1.
\end{equation}

Altogether there are 9 irreducible representations, so both  $t_{AB}$
and $\bar t_{AB}$ can have at most 9 different eigenvalues. Of course,
there can be some further degeneracies.

We compute the eigenvalues by focusing on the highest weight
states. If there is a
representation in the Clebsch-Gordan series with multiplicity one, then the
highest weight states have to be eigenstates, and the eigenvalue is
found simply by acting with the matrix on the given highest weight state.
For the representations with non-trivial multiplicities we have to
perform an explicit diagonalization on the finite
set of highest weight vectors for that given representation. This will be
detailed in the following.

The global $GL(3)$ generators are
\begin{equation}
  \Lambda_{jk}=E^{(1)}_{jk}-E^{(2)}_{kj}+E^{(3)}_{jk}-E^{(4)}_{kj}.
\end{equation}
This follows from the conjugation symmetry.
Accordingly, the highest weight vectors in the individual factors in the
tensor product \eqref{haha} can be chosen as
\begin{equation}
  \ket{1}_1,\quad \ket{3}_2,\quad \ket{1}_3,\quad \ket{3}_4.
\end{equation}

Now we consider all components in the Clebsch-Gordan series separately.

\begin{itemize}
\item $\yng(4,2)$ with dimension 27.

  This representations has multiplicity one, so the highest weight vector has to be an eigenvector
  of the transfer matrices.
  The highest weight vector is
    \begin{equation}
    \ket{1313}.
  \end{equation}

\item $\yng(3)$ with dimension 10. Once again, this representation is
  multiplicity free. The highest weight state is
  \begin{equation}
    \ket{1213}-\ket{1312}.
  \end{equation}
  \item $\yng(3,3)$ with dimension 10, multiplicity free. The highest
    weight state is:
  \begin{equation}
    \ket{1323}-\ket{2313}.
  \end{equation}
\item 4 copies of $\yng(2,1)$. Here we have 4 highest weight
  vectors, and the transfer matrices are closed in the
  subspace formed by the 4 vectors.
 A basis in this 4 dimensional space can be formed by taking one delta-state for one
  product $3\otimes\bar 3$ and taking the
  highest weight vector for the rep  $\yng(2,1)$ in the other product
  $3\otimes\bar 3$. There are indeed 4 ways to
  do this.

 The basis is thus:
  \begin{equation}
    \begin{split}
&v_1=(\ket{11}+\ket{22}+\ket{33})      \otimes \ket{13}\\
& v_2=\ket{13}\otimes  (\ket{11}+\ket{22}+\ket{33}) \\
&v_3= \ket{1311}+ \ket{2312}+ \ket{3313}\\
& v_4=\ket{1113}+\ket{1223}+\ket{1333}
\end{split}
\end{equation}

Note that this is not an orthonormal basis. The matrix elements of
$t_{AB}$ and $\bar t_{AB}$ in this space can be computed by taking into account
also the scalar products between the basis vectors. Let
$G_{ij}=\skalarszorzat{v_i}{v_j}$. Then we have explicitly
\begin{equation}
  G=
  \begin{pmatrix}
    3 & 0 & 1 & 1 \\
    0 & 3 & 1 & 1 \\
    1 & 1 & 3 & 0 \\
    1 & 1 & 0 & 3
  \end{pmatrix}.
\end{equation}
The actual matrix elements can be computed using the inverse $G^{-1}$,
so that within this space we have a 4x4 matrix $\tilde T$ such that
\begin{equation}
  \tilde T_{ij}=G^{-1}_{ik}\bra{v_k}T\ket{v_j}.
\end{equation}
We need to compute the eigenvalues of $\tilde T$, which can be done using for
example \texttt{Mathematica}.

\item Finally, there are two singlet representations. A basis is
  obtained by taking the delta-state  $\ket{\delta_{12}}\otimes
  \ket{\delta_{34}}$ and its permutation, formally written as 
$\ket{\delta_{14}}\otimes \ket{\delta_{23}}$. The first delta state is
an eigenstate, but the permuted one is not. A further simple
diagonalization is needed to find the second eigenvector as a linear
combination
\begin{equation}
  \ket{\delta_{12}}\otimes  \ket{\delta_{34}}+\alpha \ket{\delta_{14}}\otimes \ket{\delta_{23}},
\qquad \alpha\in\complex.
\end{equation}
\end{itemize}

The eigenvalues obtained with this method are listed in Table
\ref{tabtab}. We now analyze the resulting rational functions, and
show that the eigenvalue 3 is indeed the dominant one in the physical
strip, for both matrices $t_{AB}$ and $\bar t_{AB}$.

\begin{table}[h]
  \centering
  \renewcommand{\arraystretch}{2}
   \Yboxdim{10pt}
  \begin{tabular}{c|c|c}
      &  $t_{AB}$ & $\bar t_{AB}$ \\ \hline\hline
$ \yng(4,2)$   & $\frac{u v (3 u v-2)}{\left(u^2+1\right) \left(v^2+1\right)} $
& $\frac{48 u^2v^2+60(u^2 +v^2)-32 u v+99}{\left(4 u^2+9\right) \left(4 v^2+9\right)}$  \\
\hline
$  \yng(3)$   &  $\frac{3 u^2 v^2}{\left(u^2+1\right) \left(v^2+1\right)}$
 & $\frac{3 \left( 16u^2v^2+20(u^2+v^2)+4 \ii(u+v)+9\right)}
   {\left(4 u^2+9\right) \left(4 v^2+9\right)}$ \\
  \hline
$ \yng(3,3)$  &  $\frac{3 u^2 v^2}{\left(u^2+1\right) \left(v^2+1\right)}$
                  &
                  $\frac{3 \left( 16u^2v^2+20(u^2+v^2)-4 \ii(u+v)+9\right)}
        {\left(4 u^2+9\right) \left(4 v^2+9\right)} $     \\
\hline
  $4\cdot\yng(2,1)$ &  $\frac{3 v^2}{v^2+1}$ &  $3-\frac{12}{4 v^2+9}$ \\
            & $\frac{3 u^2}{u^2+1}$ & $3-\frac{12}{4 u^2+9}$ \\
            & $\frac{3 u v (u v+1)}{\left(u^2+1\right) \left(v^2+1\right)}$ &
$ \frac{3 \left(16 u^2v^2+20(u^2+v^2)+16 u v-8 \sqrt{4-5 (u-v)^2}+29\right)}{\left(4 u^2+9\right) \left(4 v^2+9\right)}$ \\ 
    & $\frac{3 u v (u v+1)}{\left(u^2+1\right)    \left(v^2+1\right)}$ &
$ \frac{3 \left(16 u^2v^2 +20(u^2+v^2)+16 u v+8 \sqrt{4-5(u-v)^2}+29\right)}
    {\left(4 u^2+9\right) \left(4 v^2+9\right)}$   \\
  \hline
  $2\cdot 1$           & 3 & 3
  \\
      &$\frac{u v (3 u v+2)}{\left(u^2+1\right) \left(v^2+1\right)} $&
       $ \frac{3 \left(16 u^2 v^2+20 (u^2+v^2)+32 u v+65\right)}{\left(4 u^2+9\right) \left(4 v^2+9\right)}$\\
  \end{tabular}
   \caption{Table of the eigenvalues of $t_{AB}$ and $\bar t_{AB}$ for
    the different representations in the Clebsch-Gordan series.}
 \label{tabtab}
\end{table}

We are interested in the eigenvalues on the physical strip
$\mathcal{P}$ with the restriction $v=u^*$. The main
idea to prove that 3 is dominant is to first consider the special point
$u=0$, for which this is seen immediately, and then to determine the
algebraic curves on which the magnitudes of the other eigenvalues reach 3. It can be
shown that all of these curves are on or outside the boundaries of the
physical strip. We will see that in both cases it is the other singlet
state which reaches the dominant eigenvalue 3 exactly on the boundary
of $\physstr$.

In the case of $t_{AB}$ six eigenvalues are of the form
$\frac{3uv(uv-C)}{(1+u^2)(1+v^2)}$, with $C=2/3, 0,-1,-2$. The
solution to the equation
\begin{equation}
 \frac{3uv(uv-C)}{(1+u^2)(1+v^2)} = 3
\end{equation}
given that $v=u^*$, $u=a+\ii b$ is 
\begin{equation}
 b=\pm \sqrt{\frac{2+C}{2-C}a^2 + \frac{1}{2-C}}.
\end{equation}
It can be seen that with these values of $C$ the minimum distance
between $u$ and the real axis is always equal to or bigger than $1/2$, thus the
eigenvalue crossings do not occur in the physical strip. The other
singlet state corresponds to $C=-2$, for which we obtain simply $b=\pm
1/2$: this state becomes degenerate with the delta state exactly at the
boundary of the physical strip. The remaining two eigenvalues of
$t_{AB}$ can be
treated in a similar way, and it can be seen that they do not become dominant
within $\physstr$.

Regarding $\bar t_{AB}$ we can see that at $u=v=0$ 3 is the dominant eigenvalue.
Once again we analyze the intersections where the magnitude of the
different eigenvalues becomes equal to 3, and we show one by one that
these curves lie outside or on the boundary of the physical strip.

\begin{itemize}
\item For the first eigenvalue in the Table \ref{tabtab} the
  substitution $u=v^*=a+\ii b$ 
   leads to the following equation for the intersection: 
 \begin{equation}
  9 + 8 a^2 - 4 b^2 = 0.
 \end{equation}
 The solution to this equation is $b=\pm \frac12\sqrt{8 a^2+9}$, which
 shows that the intersection is outside of $\physstr$. 

\item The second and third eigenvalues are a complex conjugate pair,
  hence it is sufficient to consider the absolute value of only one of
  them.  Substituting $u=a+\ii b,\ v=a-\ii b$, considering the
  magnitude of the eigenvalue and setting it equal to 3
  leads to a 6th order polynomial equation in
  $b$. However, it only contains even powers of $b$, hence leading to
  the following 3rd order equation of the intersection in $B=b^2$: 
 \begin{equation}
  64 a^6+(368+64B) a^4+(620-64B^2-160B) a^2+405-64B^3+368B^2-684B=0.
 \end{equation}
It can be seen that within the physical strip with $0\le B\le 1/4$ all
coefficients of the various powers of $a^2$ are strictly positive,
therefore there is no intersection within the physical strip.
 \item The 4th and 5th eigenvalues also form a complex conjugate pair,
   hence we only consider the first one. The direct substitution
   $u=v^*=a+\ii b$ shows that the magnitudes of these eigenvalues are
   smaller than 3 for $|b|\le 1/2$.
 \item The 6th and 7th eigenvalues on the list  are not complex
   conjugate pairs, but they have a quite similar structure, and we
   discuss them together. Substitution and simple algebraic
   manipulation leads to the following equations: 
 \begin{equation}
  4 a^2+13 =12b^2\pm4 \sqrt{5 b^2+1}.
 \end{equation}
It can be seen that the r.h.s. is always smaller than the minimum of
the l.h.s. given by 13, if $|b|\le 1/2$. Thus there is no intersection
in $\physstr$.

\item Finally, regarding the 9th eigenvalue a direct computations
  shows that this becomes degenerate with 3 if $b=1/2$, i.e. just on
  the boundary of the physical strip.
\end{itemize}
With this we have proven that on the physical strip $3$ is indeed
a non-degenerate dominant eigenvalue for both $t_{AB}$ and $\bar{t}_{AB}$.

\section{TBA derivation of the string-charge relations}

\label{sec:tbasc}

Here we compute a compact form for the mean values of the operators
$X_m(u)$ and $Y_m(u)$. We start from the expressions \eqref{XYkoztes1}.

We will use the following convention for the Fourier transform:
\begin{align}
 \hat{f}(k) &= \int_{-\infty}^\infty du\, f(u) e^{\ii ku} \\
 f(u) &= \frac{1}{2\pi} \int_{-\infty}^\infty dk\, \hat{f}(k) e^{-\ii ku}.
\end{align}
For the elementary function $a_n(\lambda)$ defined in \eqref{andef} we have
\begin{equation}
\hat{a}_n(k) = e^{-\frac{n}{2}|k|}.
\end{equation}
We need to express
\begin{align}
 \lim_{\text{TDL}} \frac{1}{L}X_m(u) &= \sum_{n=1}^\infty \int_{-\infty}^\infty dk\, \hat\rho^{(1)}_n (k) e^{-\ii ku} \sum_{j=1}^{\min(n,m)} \hat{a}_{|n-m|-1+2j}(k) \\
 \lim_{\text{TDL}} \frac{1}{L}Y_m(u) &= \sum_{n=1}^\infty \int_{-\infty}^\infty dk\, \hat\rho^{(2)}_n (k) e^{-\ii ku} \sum_{j=1}^{\min(n,m)} \hat{a}_{|n-m|-1+2j}(k).
\end{align}
We apply the summation of the geometric series
\begin{equation}
\begin{split}
 \sum_{j=1}^{\min(m,n)}& \hat{a}_{|n-m|-1+2j}(k) = \sum_{j=1}^{\min(m,n)} e^{-\frac{|k|}{2}(|n-m|-1+2j)} = \\ =& \frac{1}{2\sinh\left(\frac{|k|}{2}\right)}\left(e^{-\frac{|k|}{2}|n-m|} - e^{-\frac{|k|}{2}(n+m)} \right).
\end{split}
\end{equation}
to simplify the previous equations to
\begin{equation}
\begin{split}
 \lim_{\text{TDL}}  \frac{1}{L}X_m(u) &= \sum_{n=1}^\infty \int_{-\infty}^\infty dk\, \hat\rho^{(1)}_n (k) e^{-\ii ku} \frac{1}{\sinh\left(\frac{|k|}{2}\right)}\left( e^{-\frac{|k|}{2}|n-m|} - e^{-\frac{|k|}{2}(n+m)} \right) \\
 \lim_{\text{TDL}}  \frac{1}{L}Y_m(u) &=  \sum_{n=1}^\infty \int_{-\infty}^\infty dk\, \hat\rho^{(2)}_n (k) e^{-\ii ku} \frac{1}{\sinh\left(\frac{|k|}{2}\right)}\left( e^{-\frac{|k|}{2}|n-m|} - e^{-\frac{|k|}{2}(n+m)} \right).
\end{split}
\end{equation}
To proceed, we consider the decoupled TBA equations \eqref{eq:root1} in Fourier space:
\begin{equation}
\begin{split}
 \hat\rho^{(1)}_{n,t}(k) &= \delta_{n,1}\hat{s}(k) + \hat{s}(k)\left( \hat\rho^{(1)}_{h,n-1}(k) + \hat\rho^{(1)}_{h,n+1}(k) \right) + \hat{s}(k)\hat\rho^{(2)}_n(k) \\
 \hat\rho^{(2)}_{n,t}(k) &= \hat{s}(k)\left( \hat\rho^{(2)}_{h,n-1}(k) + \hat\rho^{(2)}_{h,n+1}(k) \right) + \hat{s}(k)\hat\rho^{(1)}_n(k),
\end{split}
\end{equation}
where
\begin{equation}
 \hat{s}(k) = \frac{1}{2\cosh\left(\frac{k}2\right)},\qquad\qquad \hat\rho^{(r)}_{h,0}(k) = 0.
\end{equation}
For simplicity we will not denote the $k$ argument in the
following.

Using the TBA equation in Fourier space, we express the
$\rho_m^{(r)},\ r=1,2$ root densities with the $\rho_{h,m}^{(r)},\ r=1,2$ hole densities:
\begin{equation}
\begin{split}
  \hat\rho^{(1)}_n &= \frac{1}{1-\hat{s}^2}\left( \hat{s}\left(\delta_{n,1}+\hat\rho^{(1)}_{h,n-1}-\frac{1}{\hat{s}}\hat\rho^{(1)}_{h,n}+\hat\rho^{(1)}_{h,n+1} \right) + \hat{s}^2 \left(\hat\rho^{(2)}_{h,n-1}-\frac{1}{\hat{s}}\hat\rho^{(2)}_{h,n}+\hat\rho^{(2)}_{h,n+1} \right) \right)\\
 \hat\rho^{(2)}_n &= \frac{1}{1-\hat{s}^2}\left( \hat{s}\left(\hat\rho^{(2)}_{h,n-1}-\frac{1}{\hat{s}}\hat\rho^{(2)}_{h,n}+\hat\rho^{(2)}_{h,n+1} \right) + \hat{s}^2\left(\delta_{n,1}+\hat\rho^{(1)}_{h,n-1}-\frac{1}{\hat{s}}\hat\rho^{(1)}_{h,n}+\hat\rho^{(1)}_{h,n+1} \right) \right).
\end{split}
\end{equation}
We make use of the following identity:
\begin{equation}
  \begin{split}
& \sum_{n=1}^\infty \left(
   \delta_{n,1}+\hat\rho^{(r)}_{n-1,h}-\frac{1}{\hat{s}}\hat\rho^{(r)}_{n,h}+\hat\rho^{(r)}_{n+1,h}
 \right) \left( e^{-\frac{|k|}2(m+n)} - e^{-\frac{|k|}2|m-n|} \right)
 =\\
 &\hspace{2cm}2\sinh\left(\frac{|k|}2\right)\left(\hat\rho^{(r)}_{m,h}-e^{-\frac{|k|}2 m}\right).
  \end{split}
\end{equation}
Substituting this and performing  some algebraic manipulations we end up with:
\begin{equation}
\begin{split}
\label{eq:SCdualityFspace}
  -\frac{1}{2\pi L}\left( 2\cosh\left(\frac k2\right) \hat{X}_m - \hat{Y}_m \right) &=  \hat\rho^{(1)}_{m,h} - e^{-\frac{|k|}{2}m} \\
  -\frac{1}{2\pi L}\left( \cosh\left(\frac k2\right) \hat{Y}_m - \hat{X}_m \right) &=  \hat\rho^{(2)}_{m,h}.
\end{split}
\end{equation}
After inverse Fourier transformation:
\begin{equation}
\begin{split}
\label{eq:SCdualityV1}
 \frac{1}{2\pi L} \left( -X_m\left(u-\frac\ii2\right) - X_m\left(u+\frac\ii2\right) +Y_m(u) \right) &= \rho^{(1)}_{m,h}(u) -a_m(u) \\
 \frac{1}{2\pi L} \left( -Y_m\left(u-\frac\ii2\right) - Y_m\left(u+\frac\ii2\right) +X_m(u) \right) &= \rho^{(2)}_{m,h}(u),
\end{split}
\end{equation}
or in a more compact form:
\begin{equation}
\begin{split}
 \rho_{h,m}^{(1)}&= a_m - \frac{1}{2\pi L}\left( X_m^{[+]} + X_m^{[-]} - Y_m \right) \\
 \rho_{h,m}^{(1)}&= - \frac{1}{2\pi L}\left( Y_m^{[+]} + Y_m^{[-]} - X_m \right).
\end{split}
\end{equation}
This is the first form of our main result, which concerns the hole densities.
It is also useful to express the root densities.

Consider \eqref{eq:SCdualityFspace} and rewrite it in matrix notation:
\begin{equation}
  \begin{pmatrix}
    2\cosh(k/2) & -1 \\
    -1 & 2\cosh(k/2)
  \end{pmatrix}
  \begin{pmatrix}
    \hat{X}_m \\ \hat{Y}_m
  \end{pmatrix}=-2\pi L
  \begin{pmatrix}
      \hat{\rho}_{h,m}^{(1)}-e^{-\frac{|k|}2 m} \\   \hat{\rho}_{h,m}^{(2)}
  \end{pmatrix}.
\end{equation}
Consider the inverse matrix, 
\begin{equation}
  \frac{1}{e^k+e^{-k}+1}
  \begin{pmatrix}
 2\cosh(k/2) &   1 \\
    1 & 2\cosh(k/2)
  \end{pmatrix}.
\end{equation}
With the help of it we can express the charges with the root densities:
\begin{equation}
  \begin{split}
  \label{eq:SCdualityV2}
    X_m&=-2\pi L\left(G_1\star  (\rho_{h,m}^{(1)}+a_m)+G_2 \star \rho_{h,m}^{(2)}\right)\\
    Y_m&=-2\pi L\left(G_1\star \rho_{h,m}^{(2)}+G_2 \star (\rho_{h,m}^{(1)}+a_m)\right),
  \end{split}
\end{equation}
where $G_1(u),\ G_2(u)$ are the inverse transformed versions of
\begin{equation}
  \begin{split}
    \hat{G}_1(k)&=\frac{e^{k/2}+e^{-k/2}}{e^k+e^{-k}+1}    \\
    \hat{G}_2(k)&=\frac{1}{e^k+e^{-k}+1},   \\
  \end{split}
\end{equation}
and their explicit form can be computed using standard techniques: 

\begin{equation}
  \begin{split}
  G_1(x)&=\int \frac{dk}{2\pi}e^{-\ii kx}\hat G_1(k)
=\frac{1}{\sqrt{3}}
\frac{\cosh(\pi/3 x)}{\cosh(\pi x)}=\frac{1}{\sqrt{3}}
   \frac{1}{2\cosh(2\pi x/3)-1} \, , \\
    G_2(x)&=\int \frac{dk}{2\pi}e^{-\ii kx}\hat G_2(k)=\frac{1}{\sqrt{3}}
    \frac{\sinh(\pi/3 x)}{\sinh(\pi x)}=\frac{1}{\sqrt{3}}
    \frac{1}{2\cosh(2\pi x/3)+1} \, .
\end{split}
\end{equation}
Consider the \eqref{eq:root1} TBA equation, and after Fourier transformation, construct a similar matrix form out of it:
\begin{equation}
 \begin{pmatrix}
  1 & -\hat{s} \\
  -\hat{s} & 1
 \end{pmatrix}
 \begin{pmatrix}
  \hat{\rho}_{t,m}^{(1)} \\ \hat{\rho}_{t,m}^{(2)}
 \end{pmatrix}
 =\begin{pmatrix}
  \delta_{m,1} \hat{s} + \hat{s}\left(\rho_{h,m-1}^{(1)}+\rho_{h,m+1}^{(1)}-\rho_{h,m}^{(2)}\right) \\
  \hat{s}\left(\rho_{h,m-1}^{(2)}+\rho_{h,m+1}^{(2)}-\rho_{h,m}^{(1)}\right)
 \end{pmatrix}.
\end{equation}
After taking the inverse matrix Fourier inverse transformation, we arrive at the following form:
\begin{equation}
  \begin{split}
     \rho_{t,m}^{(1)}&=\delta_{m,1}G_1
     +G_1\star (\rho_{h,m-1}^{(1)}+\rho_{h,m+1}^{(1)}-\rho_{h,m}^{(2)})
     +G_2\star (\rho_{h,m-1}^{(2)}+\rho_{h,m+1}^{(2)}-\rho_{h,m}^{(1)})
     \\
     \rho_{t,m}^{(2)}&=\delta_{m,1}G_2
     +G_1\star  (\rho_{h,m-1}^{(2)}+\rho_{h,m+1}^{(2)}-\rho_{h,m}^{(1)})
     +G_2\star (\rho_{h,m-1}^{(1)}+\rho_{h,m+1}^{(1)}-\rho_{h,m}^{(2)}).
 \end{split}
\end{equation}
Using  \eqref{eq:SCdualityV2} and \eqref{eq:SCdualityV1} we get
\begin{equation}
  \begin{split}
     \rho_{m}^{(1)}&=\delta_{m,1}G_1
-\frac{1}{2\pi L}\left(X_{m-1}+X_{m+1} - X_m^{[+]}-X_m^{[-]}\right)-G_1\star
(a_{m-1}+a_{m+1})+G_2\star a_m
-a_m
     \\
     \rho_{m}^{(2)}&=\delta_{m,1}G_2
     -\frac{1}{2\pi L}\left(Y_{m-1}+Y_{m+1}
   - Y_m^{[+]}-Y_m^{[-]} \right)
     -G_2\star(a_{m-1}+a_{m+1})+G_1\star a_{m}.
   \end{split}
\end{equation}
Making use of the identities:
\begin{equation}
  \begin{split}
s\star(    a_{n-1}+a_{n+1})&=a_n,\qquad n>1\\
s\star a_2&=a_1+s \\
G_2&=G_1\star s \\
-G_1\star a_2+G_2\star a_1+a_1&=-G_1
\end{split}
  \end{equation}
we express the string-charge relations in a uniform, source term free way:
\begin{equation}
  \begin{split}
  \label{eq:SU3SCdualityV3a}
     \rho_{m}^{(1)}&=\frac{1}{2\pi L}\left(X_m^{[+]} + X_m^{[-]} -X_{m-1}-X_{m+1} \right)
     \\
     \rho_{m}^{(2)}&=\frac{1}{2\pi L}\left(
  Y_m^{[+]}+Y_m^{[-]}-       Y_{m-1}-Y_{m+1}
     \right).
 \end{split}
\end{equation}

\section{Checking the string-charge relations for a particular quench}

\label{sec:deltastate}

Let us consider the quantum quench in the $SU(3)$-invariant model 
with the specific initial state
\begin{equation}
 \ket{\Psi_\delta}=\prod_{j=1}^{L/2} \frac{\ket{11}+\ket{22}+\ket{33}}{\sqrt{3}}.
\end{equation}
This quench has been treated in detail in the works
\cite{sajat-su3-1,sajat-su3-2} using Boundary Quantum Transfer Matrix
techniques. In particular, the following exact results were computed
there.

In the long time limit the system is populated by Bethe states with
root densities $\{\rho^{(1)}_m(u),\rho^{(2)}_m(u)\}$, such that the
total densities (sums of the densities of roots and holes)
$\rho^{(a)}_{t,m}(u)=\rho^{(a)}_m(u)+\rho^{(a)}_{h,m}(u)$ are
\begin{equation}
  \begin{split}
    \rho^{(1)}_{t,1}(u)&=\frac{1}{2\pi}\frac{16(80u^4+168u^2+53)}{(4u^2+1)(8u^2+3)(4u^2+9)^2}\\
    \rho^{(2)}_{t,1}(u)&=\frac{1}{2\pi}\frac{(4u^2+1)(5u^4+18u^2+8)}{(u^2+1)^2(u^2+4)^2(8u^2+3)}.
      \end{split}
\end{equation}
The ratios of root and hole densities are given by the so-called $Y$
functions defined as
\begin{equation}
  \eta^{(a)}_m(u)=\frac{\rho^{(a)}_{h,m}(u)}{\rho^{(a)}_{m}(u)}.
\end{equation}
We use this notation in order to avoid confusion with the $Y(u)$
charge operators.

The exact $Y$-functions for the 1-strings of the first and second
type were computed as
\begin{equation}
  1+\eta_1^{(1)}(u)=  1+\eta_1^{(2)}(u)=
  \frac{3(4u^2+1)}{4u^2}.
\end{equation}
From the above equations we can compute the first two hole densities as
\begin{equation}
  \label{deltarhoh}
  \begin{split}
  \rho_{h,1}^{(1)}&=\frac{ \eta_1^{(1)}(u)}{1+ \eta_1^{(1)}(u)}  \rho_{1}^{(1)}
  =\frac{1}{2\pi}\frac{16(80u^4+168u^2+53)}{3(4u^2+1)^2(4u^2+9)^2}    \\
   \rho_{h,1}^{(2)}&=\frac{ \eta_1^{(2)}(u)}{1+ \eta_1^{(2)}(u)}  \rho_{1}^{(2)}=
\frac{1}{2\pi}\frac{(5u^4+18u^2+8)}{3(u^2+1)^2(u^2+4)^2}.
 \end{split}
\end{equation}

In the following we compute the mean values of the charges
$X_1(u)$ and $Y_1(u)$ in $\ket{\Psi_\delta}$ using their definition:
\begin{equation}
  \begin{split}
    X_1(u)&=-\ii\left.\frac{\partial}{\partial \lambda}
  \bra{\Psi_\delta}\bar t(-u)t(\lambda)\ket{\Psi_\delta} \right|_{u=\lambda}\\
 Y_1(u)&=-\ii \left.\frac{\partial}{\partial \lambda}
  \bra{\Psi_\delta}(\bar t^{\bar 3}(-u)t^{\bar 3}(\lambda)\ket{\Psi_\delta} \right|_{u=\lambda}.
\end{split}
\end{equation}
From this we will compute the hole densities directly from the string-charge relations
\eqref{eq:SCsu3toki}, which will be compared to \eqref{deltarhoh}.

The above mean values can be evaluated using standard methods, by
building the corresponding 2D partition functions, and evaluating them
with double row
transfer matrices in the crossed channel
\cite{essler-xxz-gge,sajat-su3-1,sajat-su3-2}, see also Fig.  \ref{fig:bqtm0}. In the TDL we have
\begin{equation}
  \begin{split}
       \bra{\Psi_\delta}\bar t(-u)t(\lambda)\ket{\Psi_\delta} &\to 
     \left(\Lambda_{\delta}^{(1)}(\lambda,u)\right)^{L/2}\\
 \bra{\Psi_\delta}(\bar t^{\bar 3}(-u)t^{\bar 3}(\lambda)\ket{\Psi_\delta}  &\to 
\left(\Lambda_{\delta}^{(2)}(\lambda,u)\right)^{L/2}.
   \end{split}
\end{equation}
where $\Lambda_{\delta}^{(1,2)}(\lambda,u)$ are the leading
eigenvalues of the corresponding double row QTM's. 

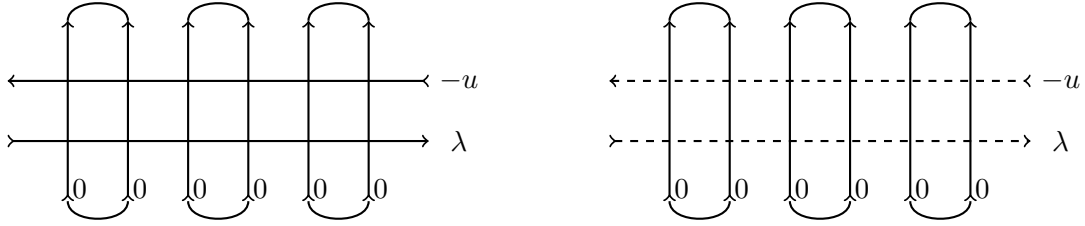
\begin{figure}
  \centering
  \begin{tikzpicture}[scale=0.8]
\foreach \x [count=\n]in {0,10}{ 
   \begin{scope}[xshift = \x cm]

     \foreach \xx [count=\n]in {0,2,4} {
       \begin{scope}[xshift = \xx cm]
   \draw [thick,>->] (0,-1) to (0,2);
\draw [thick,>->] (1,-1) to (1,2);
 \draw [thick] (0,-1) [out=-90,in=-90] to (1,-1);
 \draw [thick] (0,2) [out=90,in=90] to (1,2);        
 \node at (1.2,-0.8) {$0$};
  \node at (0.2,-0.8) {$0$};
\end{scope} 
 }
 \node at (6.5,1) {$-u$};
 \node at (6.5,0) {$\lambda$};
          \end{scope}
    }

    \node at (-2,2.5) {};
      \node at (6,2.5) {};
  \draw [thick,>->] (-1,0) to (6,0);
   \draw [thick,<-<] (-1,1) to (6,1);

    \draw [thick,dashed,>->] (9,0) to (16,0);
   \draw [thick,dashed,<-<] (9,1) to (16,1);
\end{tikzpicture}
\caption{Evaluation of the mean values of $X_1(u)$ and $Y_1(u)$ in the
  delta state. The solid and dashed lines stands for auxiliary spaces carrying the  
 defining and conjugate representations of $SU(3)$, respectively, and
 the arrows denote the direction of the action of the $R$ matrices.
 The horizontal lines with spectral parameters $\lambda$ and $u$ stem from the action of the various transfer
 matrices. 
 The vertical lines
 correspond to the physical spaces of the homogeneous chain, thus
 their spectral parameter is equal to zero.
 We have periodic boundary conditions in the horizontal direction.
  These partition functions can be evaluated in the crossed channel,
  by building the QTM's which act from the left to the right.
  }
  \label{fig:bqtm0}
\end{figure}

\begin{figure}
  \centering
  \begin{tikzpicture}[scale=0.8]
\foreach \x [count=\n]in {0,8}{ 
   \begin{scope}[xshift = \x cm]
 
   \draw [thick,>->] (0,-1) to (0,2);
\draw [thick,>->] (1,-1) to (1,2);
\draw [thick] (2,0) [out=0,in=0] to (2,1);
\draw [thick] (-1,0) [out=180,in=180] to (-1,1);
 \draw [thick] (0,-1) [out=-90,in=-90] to (1,-1);
 \draw [thick] (0,2) [out=90,in=90] to (1,2);
 \node at (2.5,1) {$-u$};
 \node at (2.5,0) {$\lambda$};
 \node at (1.2,-0.8) {$0$};
  \node at (0.2,-0.8) {$0$};
          \end{scope}
    }

    \node at (-2,2.5) {$A(\lambda,u):$};
      \node at (6,2.5) {$B(\lambda,u):$};
  \draw [thick,>->] (-1,0) to (2,0);
   \draw [thick,<-<] (-1,1) to (2,1);

    \draw [thick,dashed,>->] (7,0) to (10,0);
   \draw [thick,dashed,<-<] (7,1) to (10,1);
\end{tikzpicture}
\caption{Evaluation of the leading eigenvalues of the QTM's, which
  belong to the delta state in the crossed channel. The eigenvalue is
  computed by sandwiching the QTM between the eigenvector from the
  left and right. We thus
  obtain the partition functions above, where the boundary conditions
  are given by the delta states in all 4 directions.
  These partition functions can then be evaluated as a single trace,
 which is obtained for example in an
    anti-clockwise manner leading to the expressions \eqref{ABu0}.}
  \label{fig:bqtm1}
\end{figure}
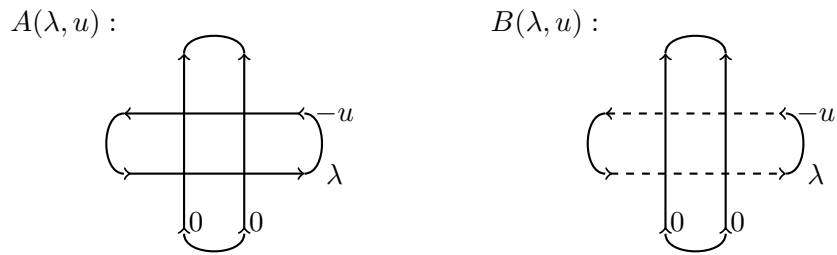

The double row QTM's can be diagonalized in a relatively simple
way. Due to the boundary conditions they are only $SO(3)$-symmetric,
and the 9 dimensional Hilbert space on which they act splits into the
$SO(3)$-representations
\begin{equation}
  3\otimes 3=5+3+1.
\end{equation}
The singlet representation corresponds to the delta state, and it
follows from the inversion relations that at $u=\lambda$ the
corresponding eigenvalue is
\begin{equation}
  \label{kkk}
  \Lambda_{\delta}^{(1)}(\lambda,\lambda)=\Lambda_{\delta}^{(2)}(\lambda,\lambda)=1.
\end{equation}
We checked that the other two eigenvalues are indeed sub-leading in
the physical strip, for both QTM's.

The mean values of the $X_1(u)$ and $Y_1(u)$ operators are thus
\begin{equation}
  \begin{split}
    X_1(u)&=-\frac{\ii L}{2}\left.\frac{\partial}{\partial \lambda}
  \Lambda_{\delta}^{(1)}(\lambda,u)
    \right|_{u=\lambda}\\
Y_1(u)&=-\frac{\ii L}{2}\left.\frac{\partial}{\partial \lambda}
 \Lambda_{\delta}^{(2)}(\lambda,u)
\right|_{u=\lambda}.
\end{split}
\end{equation}
Even though the simple result \eqref{kkk} holds at $u=\lambda$, the
leading eigenvalue is some non-trivial rational function at $u\ne
\lambda$, which we now compute.
The simplest way is perhaps to draw the partition function
corresponding to the eigenvalue, and to evaluate it as a single trace, see
Fig. \ref{fig:bqtm1}. This leads to 
\begin{equation}
  \label{ABu0}
  \begin{split}
 X_1(u)&=\left. -\frac{\ii L}{2} \frac{1}{9}\partial_\lambda 
  \text{Tr}\left(
R^t(\lambda)R^t(-u)R(-u)R(\lambda)
  \right) \right|_{u=\lambda}   \\
Y_1(u)&=\left. -\frac{\ii L}{2}\frac{1}{9} \partial_\lambda 
  \text{Tr}\left(
\bar R^t(\lambda)\bar R^t(-u)\bar R(-u)\bar R(\lambda)
  \right)    \right|_{u=\lambda}.
\end{split}
\end{equation}
Here the extra normalization factor of $1/9$ comes from the
normalization of the physical delta states.

Using  the definitions \eqref{Rdef}-\eqref{eq:R3barDef}
we have to compute the quantities
\begin{equation}
  \label{ABu}
  \begin{split}
  A(\lambda,u)&\equiv     \text{Tr}\left(
R^t(\lambda)R^t(-u)R(-u)R(\lambda)
\right)=\\
&=\frac{\text{Tr}\left(
 (\lambda+\ii K)  (-u+\ii K)
    (-u+\ii P) (\lambda+\ii P)
  \right)}
{(\lambda+\ii)^2(-u+\ii)^2}\\
  B(\lambda,u)&\equiv     \text{Tr}\left(
\bar R^t(\lambda)\bar R^t(-u)\bar R(-u)\bar R(\lambda)
\right)=\\
&=\frac{\text{Tr}\left(
    (\lambda+\ii \frac{3}{2}-\ii P)
    (-u+\ii \frac{3}{2}-\ii P)
    (-u+\ii \frac{3}{2}-\ii K)
    (\lambda+\ii \frac{3}{2}-\ii K)
  \right)}
{(\lambda+\ii \frac{3}{2})^2(-u+\ii \frac{3}{2})^2}.
  \end{split}
 \end{equation}
Direct computation of the traces gives
\begin{equation}
  A(\lambda,u)=3\frac{
3 - (4\ii)\lambda - \lambda^2  + (4\ii) u + 8 \lambda u - (2\ii) \lambda^2  u - u^2  + (2\ii) \lambda u^2  + 3 \lambda^2  u^2
  }
  {(\lambda+\ii)^2(-u+\ii)^2}
\end{equation}
and
\begin{equation}
  B(\lambda,u)=\frac{3}{16}\frac{
75 - (188\ii) \lambda - 172 \lambda^2  - (172\ii) u - 304 \lambda u +
(176\ii) \lambda^2  u - 76 u^2  + (112\ii) \lambda u^2  + 48 \lambda^2
u^2 
  }
  {(\lambda+\ii \frac{3}{2})^2(-u+\ii\frac{3}{2})^2}.
\end{equation}
For the derivatives we get
\begin{equation}
  \begin{split}
X_1(\lambda)&= \left. \frac{-\ii L}{18}\partial_\lambda A(\lambda,u)\right|_{u=\lambda}
=\frac{1}{3}\frac{1+2\lambda^2}{(\lambda^2+1)^2}   \\
Y_1(\lambda)&= \left. \frac{-\ii L}{18}\partial_\lambda B(\lambda,u)\right|_{u=\lambda}
= \frac{4}{3}\frac{5+4\lambda^2}{(4\lambda^2+9)^2}.   \\
\end{split}
 \end{equation}
And finally, we compute the one-string hole densities via
\begin{equation}
   \label{rhohnekd}
  \begin{split}
\rho_{h,1}^{(1)}&=  \frac{1}{2\pi L}\left[  X_1^{[+]}+X_1^{[-]}- Y_1\right]-a_1\\
\rho_{h,1}^{(2)}&=  \frac{1}{2\pi L}\left[  Y_1^{[+]}+ Y_1^{[-]}-X_1\right],
  \end{split}
\end{equation}
where
\begin{equation}
  a_1=\frac{1}{2\pi}\frac{1}{u^2+1/4}.
\end{equation}
After substitution we get the same results \eqref{deltarhoh} as obtained previously.

It is important that our check is independent from the derivation of
\cite{sajat-su3-1,sajat-su3-2}, which was built on the fusion
hierarchy of the Boundary QTM's. Even though the methods of \cite{sajat-su3-1,sajat-su3-2}
also involved double row, two-site transfer matrices, a close
inspection shows that the actual construction there is different, for
example the rapidities involved are chosen in a different
way. Our present check is thus
an independent confirmation of the string-charge relations in the
$SU(3)$-symmetric chain.

\section{Proof of the local inversion relations}

\label{sec:Rinv}

Here we perform an explicit computation of the product
\begin{equation}
  R^\Lambda(u)R^\Lambda(-u),
\end{equation}
where $\Lambda$ is an irreducible representation of $GL(N)$ described
by a rectangular Young diagram, and the
$R$-matrix is given generally as
\begin{equation}
  R^\Lambda(u)=\frac{u+\ii\alpha+\ii E_{ij}\Lambda_{ji}}{u+\ii\alpha'}
\end{equation}
with some shift parameters $\alpha,\alpha'\in\valos$.

The $R$-matrix acts on the tensor product of two representations. Let
us consider the Clebsch-Gordon series
\begin{equation}
  \label{CG1}
  \Lambda_1\otimes \Lambda=\oplus_k \Lambda_k,
\end{equation}
where $\Lambda_1$ is the defining 
representation. Generally the $R$-matrix can be decomposed as
\begin{equation}
  R^\Lambda(u)=\sum_k \rho_k(u)P_k.
\end{equation}
where $P_k$ are the projectors onto the invariant subspaces. It
follows that the inversion relation can be satisfied if
\begin{equation}
 \rho_k(u)\rho_k(-u)=1
\end{equation}
holds for all components. 

In order to compute $\rho_k(u)$ we first
compute the eigenvalues of $ E_{ij}\Lambda_{ji}$ which can be
expressed using quadratic Casimir operators
\begin{equation}
C_2=  \Lambda_{ij}\Lambda_{ji}
\end{equation}
as
\begin{equation}
  E_{ij}\Lambda_{ji}=\frac{1}{2}\left(
  (E_{ij}+\Lambda_{ij}) (E_{ji}+\Lambda_{ji})-E_{ij}E_{ji}-\Lambda_{ij}\Lambda_{ji}
\right).
\end{equation}

The quadratic Casimir for the representation $\Lambda$ of $GL(N)$ with
highest weight $(h_1,\dots,h_N)$ is
\begin{equation}
C_2=  \sum_{j=1}^N h^2_j+\sum_{j<k} (h_j-h_k).
\end{equation}
We consider some examples of this formula. For the defining
representation the highest weight is $(1,0,\dots,0)$ and the Casimir is 
\begin{equation}
C_2=N.
\end{equation}
For the symmetrically fused representation with highest weight
$(m,0,0,\dots)$ (corresponding to the $(1\times m)$ Young diagram)
\begin{equation}
C_2= m(m+N-1).
\end{equation}
For the anti-symmetric tensor with highest weight $(1,1,0,\dots)$  (corresponding to the $(2\times 1)$ Young diagram)
\begin{equation}
C_2=2N-3.
\end{equation}
For symmetrically fused anti-symmetric representations with highest
weight $(m,m,0,\dots)$  (corresponding to the $(2\times m)$ Young diagram)
\begin{equation}
 C_2=2m(m+N-2).
\end{equation}

Let us now focus on the $GL(3)$ representations described by
rectangular Young diagrams.
\begin{itemize}
\item Let us take $(m,0,0)$. The Clebsch Gordan series
  \eqref{CG1} has two terms:
  \begin{equation}
       (1,0,0)\otimes (m,0,0)=(m+1,0,0)\oplus(m,1,0).
  \end{equation}

  For the component  $(m+1,0,0)$  we have
  \begin{equation}
     E_{ij}\Lambda_{ji}=    m.
   \end{equation}
   For the component $(m,1,0)$ we have
   \begin{equation}
     E_{ij}\Lambda_{ji}=-1.
     \end{equation}

It follows that the eigenvalues of the $R$-matrix
\begin{equation}
R^{(m,0)}(u)=  \frac{ u-\ii\frac{m-1}{2}+\ii E_{ij}\Lambda_{ji}}{u+\ii\frac{m+1}{2}}  
\end{equation}
are
\begin{equation}
  \frac{u-\ii\frac{m-1}{2}+\ii m}{u+\ii\frac{m+1}{2}}=1,\qquad
  \frac{ u-\ii\frac{m-1}{2}-\ii}{u+\ii\frac{m+1}{2}}=  \frac{ u-\ii\frac{m+1}{2}}{u+\ii\frac{m+1}{2}}.
\end{equation}
\item Let us now take $(m,m,0)$. There are two components in
  the Clebsch-Gordan series \eqref{CG1}:
  \begin{equation}
       (1,0,0)\otimes (m,m,0)=(m+1,m,0)\oplus (m,m,1).
  \end{equation}

For the component  $(m+1,m,0)$  we have
  \begin{equation}
    \begin{split}
      E_{ij}\Lambda_{ji}&=    m,
      \end{split}
  \end{equation}
  whereas for $(m,m,1,0)$ we have
  \begin{equation}
    \begin{split}
   E_{ij}\Lambda_{ji}&=    -2.
    \end{split}
  \end{equation}
  It follows that the eigenvalues  of
  \begin{equation}
  R^{(1,0),(0,m)}(u)=  \frac{ u-\ii\frac{m-2}{2}+\ii E_{ij}\Lambda_{ji}}{u+\ii\frac{m+2}{2}}
\end{equation}
are
\begin{equation}
  \frac{ u-\ii\frac{m-2}{2}+\ii m}{u+\ii\frac{m+2}{2}}=1,\qquad
   \frac{ u-\ii\frac{m-2}{2}-2\ii}{u+\ii\frac{m+2}{2}}=  \frac{ u-\ii\frac{m+2}{2}}{u+\ii\frac{m+2}{2}}.
 \end{equation}
\end{itemize}
We can see that both $R$-matrices satisfy the local inversion relation.

\addcontentsline{toc}{section}{References}
\providecommand{\href}[2]{#2}\begingroup\raggedright\endgroup


\end{document}